\documentclass[11pt]{article}
\pdfoutput=1
\usepackage[utf8]{inputenc}
\usepackage[margin=0.75in]{geometry}
\usepackage{graphicx, lscape}
\usepackage{graphics, color}
\usepackage{caption}
\usepackage{subcaption}
\usepackage{atbegshi}
\AtBeginDocument{\AtBeginShipoutNext{\AtBeginShipoutDiscard}}

\usepackage{fancyhdr}
\usepackage{answers}
\usepackage{setspace}
\usepackage{graphicx}
\usepackage{enumitem}
\usepackage{multicol}
\usepackage{mathrsfs}
\usepackage{amsmath,amsthm,amssymb,amscd}
\usepackage[cjk]{kotex}

\usepackage{url}
\usepackage[hidelinks]{hyperref}
\usepackage{cleveref}

\newtheorem{proposition}{Proposition}[]

\newtheorem{theorem}{Theorem}[]
\newtheorem{lemma}{Lemma}[]

\theoremstyle{definition}
\newtheorem{definition}{Definition}[]
\newtheorem{remark}{Remark}[]
\newtheorem{example}{Example}[]
\usepackage{fancyhdr}
\usepackage{mathtools}
\usepackage{algorithm}
\usepackage{algpseudocode}
\usepackage{fmtcount}
\usepackage{graphicx}
\usepackage{subcaption}
\usepackage{afterpage}
\usepackage[noadjust]{cite}
\usepackage[font=small,labelfont=bf]{caption}
\usepackage[labelfont=bf]{caption}
\usepackage{setspace}
\allowdisplaybreaks
\usepackage[title,titletoc]{appendix}
\usepackage{epsfig}
\usepackage{graphicx}
\usepackage{hhline}
\usepackage{latexsym}
\usepackage{multirow} 
\usepackage{array}
\usepackage{caption}
\usepackage{subcaption}
\usepackage{color}
\usepackage{natbib}
\usepackage{rotating}
\usepackage{graphicx,lscape}
\usepackage{graphicx}
\usepackage{epsfig}
\usepackage{graphics,color, graphicx}
\usepackage{latexsym}
\usepackage{fullpage}
\usepackage{booktabs,fixltx2e}
\usepackage[flushleft]{threeparttable}
\usepackage{fancyvrb}
\usepackage{pdfpages}
\usepackage{fmtcount}
\usepackage[symbol]{footmisc}

\newcommand{\indep}{\perp \!\!\! \perp}

\usepackage{titlesec}

\begin{document}
	\begin{center}
    	\clearpage
        \pagenumbering{arabic} 
        
        {\Large{\textbf{Absolute average and median treatment effects as causal estimands on metric spaces}}\medskip}
        \vskip 7mm
        {\sc Ha-Young Shin, Kyusoon Kim\footnote{Corresponding author: kyu9510@snu.ac.kr}, Kwonsang Lee, and Hee-Seok Oh}\\
{Department of Statistics, Seoul National University, Seoul 08826, Korea}
\end{center}
\vskip 5mm
     
    \noindent 
{\bf Abstract}: We define the notions of absolute average and median treatment effects as causal estimands on general metric spaces such as Riemannian manifolds, propose estimators using stratification, and prove several properties, including strong consistency. In the process, we also demonstrate the strong consistency of the weighted sample Fr\'echet means and geometric medians. Stratification allows these estimators to be utilized beyond the narrow constraints of a completely randomized experiment. After constructing confidence intervals using bootstrapping, we outline how to use the proposed estimates to test Fisher's sharp null hypothesis that the absolute average or median treatment effect is zero. Empirical evidence for the strong consistency of the estimators and the reasonable asymptotic coverage of the confidence intervals is provided through simulations in both randomized experiments and observational study settings. We also apply our methods to real data from an observational study to investigate the causal relationship between Alzheimer's disease and the shape of the corpus callosum, rejecting the aforementioned null hypotheses in cases where conventional Euclidean methods fail to do so. Our proposed methods are more generally applicable than past studies in dealing with general metric spaces. 
\vskip 5mm
\noindent {\it \bf Keywords}: Causal inference; Geometric statistics; Manifold statistics; Metric space; Stratification; Treatment effect.

\doublespacing
\section{Introduction}

The need for causal inference beyond the real line is growing as new types of non-scalar data become available. There have been many attempts to perform causal inference for data sets where the outcome variables lie in non-standard spaces. Several authors have conducted studies on causal inference for high-dimensional data. \cite{Schaechtle2013} presented a method called multi-dimensional causal discovery for uncovering causal relations within high-dimensional data like multivariate time series. \cite{Wahl2023} proposed a constraint-based non-parametric approach to studying the causal relationship between two vector-valued random variables from observational data. Spatial data is also a popular subject for causal inference; unique properties of spatial processes, such as spatial dependence and spatial heterogeneity, can lead to indirect effects and violations of fundamental assumptions of prior existing frameworks for causal inference, resulting in inaccurate estimation of the causal effect \citep{Akbari2023}. \cite{Delgado2015} and \cite{Zhang2019} measured indirect effects, called spatial spillover effects, and \cite{Arpino2016} applied an appropriate transformation to data to cope with the violation of the stable unit treatment value assumption. \cite{Giffin2022} proposed a method based on a generalized propensity score to deal with direct and indirect effects for spatial processes. \cite{Akbari2023} provided an overview of various methods in spatial causal inference. Other researchers have expanded the realm of causal inference into non-Euclidean spaces.  \cite{Wein2021} presented a data-driven model using a graph neural network framework to infer causal dependencies in brain networks. \cite{Lin2023} suggested a causal inference framework for distribution functions that reside in a metric space called Wasserstein space. \cite{Ogburn2022} illustrated a semiparametric estimation of causal effects for data from a single social network. 

These past studies have focused on specific data types like high-dimensional Euclidean data, networks, and distributions. This study generalizes beyond this past research, dealing with data in metric spaces and geodesic spaces, including Riemannian manifolds. There has been research that has briefly dealt with two-sample inference on Riemannian manifolds, such as \cite{Hendriks1998}, \cite{Bhattacharya2005} and \cite{Bhattacharya2012}; two-sample testing can be useful as a tool in causal inference, but only under highly limited conditions such as completely randomized experiments.

We define new causal estimands on metric spaces, called the absolute average treatment effect ($\textsc{aate}$) and the absolute median treatment effect ($\textsc{amte}$). We propose estimators for $\textsc{aate}$ and $\textsc{amte}$ and demonstrate that on complete connected Riemannian manifolds, these estimators can be defined using simple geodesic regression. Crucially, these estimators employ stratification and so can be used far more broadly than the two-sample inference mentioned in the previous paragraph. We investigate the strong consistency of the estimators, and of weighted sample Fr\'echet means and geometric medians, under certain conditions, and suggest testing procedures and bootstrap confidence intervals. We conduct simulation studies to examine the finite-sample performance for consistency of the estimators and the confidence interval coverage in both randomized experiments and observational study designs. We explore the causal relationship between Alzheimer’s disease and the shape of the corpus callosum, which lies on a Riemannian manifold called Kendall’s shape space, with our estimators and confidence intervals, using matched data from an observational study, and observe that our methods seem to work significantly better than conventional Euclidean methods. The data and R code used in the experiments are available at \url{https://github.com/qsoon/AATE-AMTE}.

        

\section{Background}
\subsection{Causal inference} \label{causal}
        
        \subsubsection{The potential outcomes framework and causal effects}
        The potential outcomes framework is the most widely used framework for causal inference in statistics and social and biomedical sciences. Given sets $\mathcal{T}$ and $\mathcal{O}$, suppose we have a treatment variable $z$ that takes values in $\mathcal{T}$ and an outcome that takes values in $\mathcal{O}$. In the potential outcomes framework, associated with a unit is a map $r:\mathcal{T}\rightarrow\mathcal{O}$ that describes all of the \textit{potential outcomes} for that unit; that is, the outcome will be recorded at $r(z)\in\mathcal{O}$ if the unit receives treatment $z\in\mathcal{T}$. Then, the observed outcome for this unit is $r(z)$, where $z$ is the treatment received. Even though $r(z)$ depends on $z$, the map $r$ itself, that is, the list of all potential outcomes, does not. However, we can only observe one of the potential outcomes, $r(z)$; all other potential outcomes are counterfactual and hence unobservable. This is the so-called \textit{fundamental problem of causal inference} (\citealp{Holland1986}); this problem can be somewhat mitigated through randomization. In this paper, we only consider binary treatment variables $z$ for which $\mathcal{T}=\{0,1\}$; $z=0$ means that the unit was in the control group and received no treatment, and $z=1$ means that the unit was treated. We then define $r_T=r(1)$ and $r_C=r(0)$, the two potential outcomes.
        
        
        In this section, we deal with the traditional case in which an outcome variable is real-valued. As described above, for each level $z$ of treatment $Z$, there is a potential outcome, $r(z)$. Then, $r(z)$ equals $zr_T + (1-z)r_C$. The overall treatment effect, also known as the causal effect, is usually defined as the difference between some statistical functional of the distributions of $r_T$ and $r_C$. Because causal estimands defined in this way compare two potential populations, one that has received the treatment and one that has not, they can be interpreted at the population level. Population-level interpretations are relevant for public policy as they describe the effect of a treatment on an entire population. The best-known measure of overall treatment effect is the average treatment effect ($\textsc{ate}$), 
    \begin{equation*}
\tau=E(r_T)-E(r_C).
    \end{equation*}
   In addition, there is the quantile treatment effect ($\textsc{qte}$), $\eta_\delta=F_{r_T}^{-1}(\delta)-F_{r_C}^{-1}(\delta)$ for $\delta\in(0,1)$, where $F_Y^{-1}(\delta)=\inf\{y:\text{pr}(Y\leq y)\geq\delta\}$ is the $\delta$th quantile of random variable $Y$. In particular, we can define the median treatment effect by letting $\delta=0.5$ as
    \begin{equation*}
    \eta_{0.5}=F_{r_T}^{-1}(0.5)-F_{r_C}^{-1}(0.5).
    \end{equation*}
    Defining the individual treatment effect for a unit as $\textsc{ite}=r_T-r_C$ yields $\tau=E(\textsc{ite})$. Therefore, $\textsc{ate}$ permits both an individual-level interpretation as the average individual treatment effect and the population-level interpretation as the difference between the averages of two hypothetical populations. In contrast, the difference between the $\delta$th quantiles of the treated and control populations is generally not the $\delta$th quantile of the individual treatment effects (\citealp{Imbens2015}); that is, in general, $F_{r_T-r_C}^{-1}(\delta)\neq F_{r_T}^{-1}(\delta)-F_{r_C}^{-1}(\delta)$. Therefore, only a population-level interpretation is possible for $\textsc{qte}$. 

    \subsubsection{Assignment mechanism} \label{confounder}
    
       Data used in causal inference is typically collected through a randomized experiment or an observational study. In a randomized experiment, the researchers decide which units receive treatment; that is, they control the so-called assignment mechanism. Suppose we have $N$ units and for $i=1,\ldots,N$, let $z_i\in\{0,1\}$, $r_{Ti}\in M$, $r_{Ci}\in M$, and $x_i\in\mathbb{R}^k$ be instances of the treatment variable, the potential treatment and control outcomes, and a vector of $k$ covariates, respectively, where potential outcomes take values in some space $M$. These can be collected into ordered $N$-tuples $Z_N=(z_1,\ldots,z_N)$, $R_{TN}=(r_{T1},\ldots,r_{TN})$, $R_{CN}=(r_{C1},\ldots,r_{CN})$, and $X_N=(x_1,\ldots,x_N)$, respectively. Then, the \textit{assignment mechanism} is simply 
    \begin{equation*}
    \text{pr}(Z_N=Z \mid R_{TN},R_{CN}, X_N)
    \end{equation*}
    for all $Z\in\{0,1\}^N$.
    
    The unit-level assignment probability for unit $i$ is 
    \begin{equation*}
    \text{pr}(z_i=1 \mid R_{TN},R_{CN}, X_N)=\sum_{Z:z_i=1}\text{pr}(Z_N=Z \mid R_{TN},R_{CN}, X_N).
    \end{equation*}
    An assignment mechanism is said to be \textit{probabilistic} if $0<\text{pr}(z_i=1 \mid R_{TN},R_{CN}, X_N)<1$ for all $i=1,\ldots,N$, and \textit{individualistic} if the unit-level assignment probability depends only on $x_i$, $r_{Ti}$, and $r_{Ci}$: $\text{pr}(z_i=1 \mid R_{TN},R_{CN}, X_N)=q(x_i,r_{Ti},r_{Ci})$ for some function $q$. An assignment mechanism is said to be \textit{unconfounded} if it does not depend on the potential outcomes: $\text{pr}(Z_N=Z \mid R_{TN},R_{CN}, X_N)=\text{pr}(Z_N=Z \mid R_{TN}',R_{CN}', X_N)$ for all $Z\in\{0,1\}^N$, $R_{TN}\in M^N$, $R_{CN}\in M^N$, $R_{CN}'\in M^N$, $R_{CN}'\in M^N$, and $X_N\in(\mathbb{R}^k)^N$. In this case, $R_{TN}$ and $R_{CN}$ can be dropped from the notation, and the assignment mechanism can be written as $\text{pr}(Z_N \mid X_N)$. A \textit{randomized experiment} is one in which the assignment mechanism is known and controlled by the researchers and is probabilistic. In a \textit{classical randomized experiment}, the assignment mechanism is also individualistic and unconfounded.

    The problem with randomized experiments is that they are often expensive, time-consuming, and, in some cases, unethical. On the other hand, observational studies are usually simpler, faster, and less expensive, so observational data is commonly used in causal inference. Since the researchers have no control over treatment assignment in an observational study, several assumptions are needed for causal inference. Most importantly, the study must be \textit{free of hidden bias}, meaning that all confounders, which are pre-treatment variables that may affect the treatment variable and/or the outcome, are contained in the covariates. The existence of unmeasured confounders can cause serious problems in treatment effect estimation, like spurious effects, and counterfactual outcome estimation, like selection bias. \textit{Spurious effects} occur when estimates for the causal estimands include not only the effects of treatment but also the effects of confounders on the outcome, and \textit{selection bias} occurs when the distributions of the covariates in the observed and interested groups are different, which makes counterfactual outcome estimation difficult (\citealp{Yao2021}).
    
    More precisely, the necessary assumptions in an observational study are the same as those in a classical randomized experiment, namely that the assignment mechanism is individualistic, probabilistic, and unconfounded, except that the assignment mechanism is not controlled by the researchers and may be unknown. Assuming the $(r_{Ti},r_{Ci},x_i)$ are independent,
    \begin{equation} \label{unconfounded}
    Z_N\indep(R_{TN},R_{CN}) \mid X_N
    \end{equation}
    in observational studies of this type. The condition (\ref{unconfounded}) also holds in classical randomized experiments.
        

        

        \subsubsection{Matching} \label{matchingmethod}

        With observational data, we may wish to approximate a classical stratified randomized experiment by stratifying the data based on their covariates. An \textit{exact stratification} places two units $i_1$ and $i_2$ in the same stratum if their covariates are equal, i.e., $x_{i_1}=x_{i_2}$. This is useful because of the assumption in (\ref{unconfounded}). It can be shown that $z_i$ and $(r_{Ti},r_{Ci})$ are independent given the so-called propensity score: $e(x)=\text{pr}(z_i=1 \mid x_i=x)$, and a stratification that places two units $i_1$ and $i_2$ in the same stratum if $e(x_{i_1})=e(x_{i_2})$ is called an \textit{exact matching}. An exact matching has the same advantages as an exact stratification. Of course, $e(x)$ is not known and needs to be approximated. In practice, exact stratification/matching is not feasible for non-categorical covariates or even a large number of categorical covariates. There are many matching methods used to approximate an exact matching, including propensity score matching, propensity score caliper matching, and HSIC-NNM (\citealp{Chang2017}). Matching methods can also be classified in terms of the number and/or ratio of treatment to control units in each matched set: pair matching, ratio matching, and full matching according to the ratio of treated and control units in each stratum.

        

\subsection{Geometric preliminaries} \label{geom}
       Consider a metric space $(M,d)$, which we refer to as $M$ without explicitly mentioning $d$, as a measurable space $(M,\mathcal{B})$, where $\mathcal{B}$ is the Borel $\sigma$-algebra induced by the metric $d$, along with a probability space $(\Omega,\mathcal{F},\text{pr})$ and an $M$-valued random element $y_0:(\Omega,\mathcal{F},\text{pr})\rightarrow(M,\mathcal{B})$. For $\alpha>0$, the $L_\alpha$ estimator set of $y_0$ is defined as the set of minimizers of
       \begin{equation} \label{frechet}
        f_{\alpha}(p)=E(d(p,y_0)^\alpha).
       \end{equation}
       When $\alpha=2$, this set is also called the Fr\'echet mean set, and when $\alpha=1$, the geometric median set. If the set contains exactly one element, that element is called the $L_\alpha$ estimator, or Fr\'echet mean or geometric median as appropriate.       
       
       A proper metric space is one whose subsets are compact if and only if they are closed and bounded, or equivalently, all closed balls are compact.

     \section{Absolute average and median treatment effects} \label{ci}

    \subsection{Definitions}

        Since general metric spaces lack some of the properties we take for granted in Euclidean space, problems arise when defining standard measures of causal effects, like the average and median treatment effects in such spaces. 
        We propose the following definitions. Here, $M$ is a metric space, $z$ is a binary treatment variable, and $r_T$ and $r_C$ are $M$-valued random elements.

    \begin{definition} 
    If $r_T$  and $r_C$ have unique Fr\'echet means $\mu_{2T}$ and $\mu_{2C}$, respectively, the absolute average treatment effect is defined to be $\textsc{aate}=d(\mu_{2C},\mu_{2T}).$ If they have unique geometric medians $\mu_{1T}$ and $\mu_{1C}$, respectively, the absolute median treatment effect is defined to be $\textsc{amte}=d(\mu_{1C},\mu_{1T})$.
    \end{definition}

   
    
    \cite{Lin2023} briefly considered $\textsc{aate}$ as a causal estimand in the case of Wasserstein space. When $M$ is a complete connected Riemannian manifold, \cite{Karcher1977}, \cite{Bhattacharya2003}, \cite{Fletcher2009}, and \cite{Yang2010} described conditions under which the Fr\'echet mean and geometric median uniquely exist.

    Because $\textsc{aate}$ and $\textsc{amte}$ compare two potential populations, one that has received the treatment and one that has not, they can be interpreted at the population level, like the quantile treatment effects in traditional causal inference. Population-level interpretations are relevant to public policy as they describe the effect of a treatment on the entire population.

    \subsection{Estimators for \textsc{AATE} and \textsc{AMTE}}
    
         Let $M$ be a metric space and $R_N(Z_N)=(r_1(z_1),\ldots,r_N(z_N))\in M^N$ consist of observed outcome variables with an associated vector of binary treatment variables $Z_N=(z_1,\ldots,z_N)\in\{0,1\}^N$ in a completely randomized experiment. Any term subscripted by $N$ is random in this paper. From here on, we suppress the dependence on $Z_N$ in the notation $R_N$ unless necessary. Then, the obvious estimator for $\textsc{aate}$ is the distance between the sample Fr\'echet means of the treated and control $r_i$. Building on this idea, we consider data with $\Xi$ strata. We now have an additional vector of observations $S_N=(s_1,\ldots,s_N)\in\{1,\ldots,\Xi\}^N$ which records stratum membership and stratum-wise weights $\hat{\Lambda}_N=(\hat{\lambda}_N^1,\ldots,\hat{\lambda}_N^\Xi)$ that satisfy $\sum_{s=1}^\Xi\hat{\lambda}_N^s=1$. Denote the number of treated and control units in stratum $s$ by $m_{TN}^s=\sum_{k=1}^Nz_k I\{s_k=s\}$ and $m_{CN}^s=\sum_{k=1}^N(1-z_k) I\{s_k=s\}$, respectively, where $I$ is an indicator function. Then,
    \begin{equation} \label{tl21}
    T_{\alpha}(Z_N,R_N,S_N)=\inf_{\bar{r}_{T}\in \bar{A}_{TN}^\alpha,\bar{r}_{C}\in\bar{A}_{CN}^\alpha}d(\bar{r}_{C},\bar{r}_{T}),~~ \quad (\alpha=1,2),
    \end{equation} 
    where $\bar{A}_{TN}^\alpha$ is the weighted sample $L_\alpha$ estimator set of $r_i$ with associated weights $w_{TiN}=\sum_{s=1}^\Xi\hat{\lambda}_N^sz_iI\{s_i=s\}/m_{TN}^s, (i=1,\ldots,N)$, and $\bar{A}_{CN}^\alpha$ is that with associated weights $w_{CiN}=\sum_{s=1}^\Xi\hat{\lambda}_N^s(1-z_i)I\{s_i=s\}/m_{CN}^s,  (i=1,\ldots,N)$. The proposed $T_{\alpha}$ is a possible estimator for $\textsc{aate}$ when $\alpha=2$ and $\textsc{amte}$ when $\alpha=1$. In both treatment and control cases, the sum of the weights of units within stratum $s$ is $\hat{\lambda}_N^s$, so the weights add up to 1. This stratification and weighting makes $T_{\alpha}$ useful in stratified randomized experiments and matched observation studies, unlike two-sample inference. We remark that if $M=\mathbb{R}$, the proposed estimator $T_{2}$ becomes $\lvert\sum_{s=1}^\Xi\hat{\lambda}_N^s(\bar{r}_{T}^s-\bar{r}_{C}^s)\rvert$, where $\bar{r}_{T}^s$ and $\bar{r}_{C}^s$ are the sample means of the treated and control outcomes in stratum $s$, respectively. In other words, in the real case, $T_{2}$ is the standard estimator for the absolute value of $\textsc{ate}$ in a classical randomized experiment with strata. 

    \begin{remark} \label{unique}
        The infimum in (\ref{tl21}) ensures that $T_{\alpha}$ is well-defined even when $\bar{A}_{CN}^\alpha$ or $\bar{A}_{TN}^\alpha$ contains more than one element. However in practice, outside of degenerate cases, the weighted sample Fr\'echet mean and geometric median sets are usually singletons. In particular, uniqueness is guaranteed on Hadamard manifolds (that is, complete, simply connected Riemannian manifolds of non-positive sectional curvature) such as hyperbolic spaces in all cases for the sample Fr\'echet mean and as long as the data are not confined to a single geodesic for the sample geometric median. See \cite{Karcher1977}, \cite{Bhattacharya2003}, \cite{Fletcher2009}, and \cite{Yang2010} for more sufficient, but not necessary, conditions under which uniqueness holds.
    \end{remark}

When $M=\mathbb{R}$ and the number of strata $\Xi=1$, it is known that $T_{2}$ is identical to $\lvert \hat{\beta}_1\rvert$, where $(\hat{\beta}_0,\hat{\beta}_1)$ is the least squares estimate for $(\beta_0,\beta_1)$ in the simple linear regression model $r_i=\beta_0+\beta_1z_i$ where the sole independent variable is treatment. Thus, even though the coefficients of linear regression ordinarily measure only association, and not causality, between the dependent and independent variables, here $\lvert\hat{\beta}_1\rvert$ also has a causal interpretation because it is identical to $T_{\alpha}$; this is true regardless of the validity of the regression model. Using the simple geodesic regression model for Riemannian manifold-valued dependent variables introduced by \cite{Fletcher2013}, one can show that a similar result, generalized to accommodate stratification, holds when $M$ is a complete connected Riemannian manifold, as detailed below.

Let $M$ be a complete connected Riemannain manifold with the same classical randomized experimental setup as in the general metric case, and associate a single weight
    \begin{equation*}
    W_{iN}=\beta_T\sum_{s=1}^\Xi\frac{\hat{\lambda}_N^sz_iI\{s_i=s\}}{m_{TN}^s}+\beta_C\sum_{s=1}^\Xi\frac{\hat{\lambda}_N^s(1-z_i)I\{s_i=s\}}{m_{CN}^s},~~\quad i=1,\ldots,N
    \end{equation*}
     to each data point, where $\beta_T,\beta_C\in(0,1)$ and $\beta_T+\beta_C=1$ to ensure that the weights add to 1. Let $\bar{H}^\alpha$ be the weighted simple geodesic regression $L_\alpha$ estimator set of the points $(z_i,r_i)\in\mathbb{R}\times M$, $i=1,\ldots,N$ with weights $W_{iN}$; see Section A.1.1 of the Supplementary Materials for an overview of geodesic regression on Riemannian manifolds. The elements of $\bar{H}^\alpha$ are of the form $(p,v)$, where $p\in M, v\in T_pM$.

        \begin{theorem} \label{l2equiv}
Let $M$ be a complete connected Riemannian manifold and $\alpha=1$ or $2$. Then, $\inf_{(\bar{p},\bar{v})\in\bar{H}^\alpha}\lVert \bar{v}\rVert=T_{\alpha}(Z_N,R_N,S_N)$ and is invariant with respect to $\beta_T$ and $\beta_C$.

    \end{theorem}

In fact, it is possible to extend this theorem to so-called geodesic spaces by making a straightforward generalization, presented in Section A.1.2 of the Supplementary Materials, of the geodesic regression model of \cite{Fletcher2013} from complete connected Riemannian manifolds to geodesic spaces.
See Section B.1 of the Supplementary Materials for a proof of this extended result, Theorem \ref{l2equiv}. Thus, in the case of complete connected Riemannian manifolds, one can equivalently define $T_{\alpha}$ in terms of simple geodesic regression, giving a formal justification for a causal interpretation of geodesic regression in this context.

    \subsection{Strong consistency of the estimators}
To prove the strong consistency of our proposed estimators, we first need the following result demonstrating the strong consistency of measurable selections from what we dub stratification-weighted sample $L_\alpha$ ($\alpha=1,2$) estimator sets on proper metric spaces. Strictly speaking, this result is not inherently related to causal inference, but it is a result that can be of independent interest.
    
    \begin{theorem} \label{l2gzation}

    Let $(y_0,s_0),(y_1,s_1),(y_2,s_2),\ldots:\Omega\rightarrow M\times\{1,\ldots,\Xi\}$ be independent but not necessarily identically distributed random elements such that $y_0 \mid (s_0=s),y_1 \mid (s_1=s),y_2 \mid (s_2=s),\ldots$ are identically distributed for $s=1,\ldots,\Xi$. Suppose that for $s=1,\ldots,\Xi$, $\hat{\lambda}_N^s$ defined on $\Omega$ converges almost surely to $\text{pr}(s_0=s)\in(0,1)$, and $m_N^s\rightarrow\infty$ almost surely, where $m_N^s=\sum_{k=1}^NI\{s_k=s\}$. Let $\alpha\in\{1,2\}$.    
    If $y_0$ has a unique $L_\alpha$ estimator, then every measurable selection from the weighted sample $L_\alpha$ estimator set of $y_1,\ldots,y_N$ with weights 
    \begin{equation} \label{win}
    w_{iN}=\sum_{s=1}^\Xi\frac{\hat{\lambda}_N^sI\{s_i=s\}}{m_N^s} \quad(i=1,\ldots,N),
    \end{equation}
    which we call a stratification-weighted sample $L_\alpha$ estimator set, converges to the $L_\alpha$ estimator of $y_0$.
    
    \end{theorem}

See Section B.2 of the Supplementary Materials for the proof. By letting $\Xi=1$, Theorem \ref{l2gzation} can be seen as a generalization of results in \cite{Bhattacharya2003}, \cite{Yang2011}, \cite{Ginestet2013}, and \cite{Evans2024}.
\vspace{3mm}

    The following is the main result of this paper.

    \begin{theorem} \label{l2sc}
    Equip $M$ with its Borel $\sigma$-algebra, the set $\{1,\ldots,\Xi\}$ with the discrete $\sigma$-algebra, and $M\times M\times\{1,\ldots,\Xi\}$ with the induced product $\sigma$-algebra. Let $(r_{T0},r_{C0},s_0),(r_{T1},r_{C1},s_1),\ldots:\Omega\rightarrow M\times M\times\{1,\ldots,\Xi\}$ be independent, but not necessarily identically distributed, random elements such that $(r_{T0},r_{C0}) \mid (s_0=s),(r_{T1},r_{C1}) \mid (s_1=s),\ldots$ are identically distributed for $s=1,\ldots,\Xi$. Define $r_i=r_i(z_i)$ by $r_i(z_i)=r_{Ti}$ if $z_i=1$ and $r_i(z_i)=r_{Ci}$ if $z_i=0$, where $z_i\in\{0,1\}$ is a binary variable for each positive integer $i$. Assume that 
    \begin{equation} \label{unconfoundedness}
    (z_1,\ldots,z_N)\indep((r_{T1},r_{C1}),\ldots,(r_{TN},r_{CN})) \mid (s_1,\ldots,s_N)
    \end{equation}
    for all $N\in\mathbb{Z}^+$. Suppose that for $s=1,\ldots,\Xi$, $\hat{\lambda}_N^s$ defined on $(\Omega,\mathcal{F},\text{pr})$ converges almost surely to $\text{pr}(s_0=s)\in(0,1)$), and $m_{TN}^s\rightarrow\infty$ almost surely and $m_{CN}^s\rightarrow\infty$ almost surely, where $m_{TN}^s=\sum_{i=1}^Nz_iI\{s_i=s\}$ and $m_{CN}^s=\sum_{i=1}^N(1-z_i)I\{s_i=s\}$. Suppose that $E(d(p^*,r_{T0})^\alpha)$ and $E(d(q^*,r_{C0})^\alpha)$ are finite for some points in $p^*,q^*\in M$ and the $L_\alpha$ estimator sets of $r_T$ and $r_C$ are singletons. As long as $T_{\alpha}(Z_N,R_N,S_N)$ is measurable, it is a strongly consistent estimator of $\textsc{aate}$ if $\alpha=2$ and of $\textsc{amte}$ if $\alpha=1$.
    \end{theorem}

A proof of Theorem \ref{l2sc} is provided in Section B.3 of the Supplementary Materials. If $(r_{T0},r_{C0},s_0)$, $(r_{T1},r_{C1},s_1),\ldots$ are independent and identically distributed, $(r_T,r_C) \mid (s_0=s),(r_{T1},r_{C1}) \mid (s_1=s),\ldots$ are also identically distributed and we can apply the above theorem. In this case, the obvious choice for $\hat{\lambda}_N^s$ that satisfies almost sure convergence is $m_N^s/N$, the proportion of units, in stratum $s$. The best choice for $\hat{\lambda}_N^s$ is probably $\text{pr}(s_0=s)$ itself if it is known.

We now demonstrate why the more general case considered in the theorem, in which $(r_{Ti},r_{Ci},s_i)$ are not necessarily independent and identically distributed, is useful, especially if each $\text{pr}(s_0=s)$ is known. Suppose a researcher wants to ensure that the number of units in each stratum exceeds a certain minimum threshold. Units that are independent and identically distributed may not achieve this, but it is easy to select units according to the conditions in the theorem and guarantee that these thresholds are exceeded.
    
    The condition (\ref{unconfoundedness}) in Theorem \ref{l2sc} for a specific $N$ implies the analogous condition for all positive integers smaller than $N$.
    See Lemma 7 in Section B.3 of the Supplementary Material. Thus, this theorem applies to unconfounded settings. Examples include classical randomized experiments and observational studies free of hidden bias with categorical covariates, where each of the finitely many possible combinations of covariates can be treated as a stratum.

     \begin{remark} \label{intuitive}
    Another intuitive estimator for $\textsc{aate}$ is
    \begin{equation} \label{wrong}
    \inf d(\bar{r}_C(\bar{r}_{C}^1,\ldots,\bar{r}_{C}^\Xi;
    \hat{\lambda}_N^1,\ldots,\hat{\lambda}_N^\Xi),\bar{r}_T(\bar{r}_{T}^1,\ldots,\bar{r}_{T}^\Xi;
    \hat{\lambda}_N^1,\ldots,\hat{\lambda}_N^\Xi)),
    \end{equation}
    where $\bar{r}_T(\bar{r}_{T}^1,\ldots,\bar{r}_{T}^\Xi; \hat{\lambda}_N^1,\ldots,\hat{\lambda}_N^\Xi)$ is a weighted sample Fr\'echet mean of $\bar{r}_{T}^1,\ldots,\bar{r}_{T}^\Xi$ with associated weights $\hat{\lambda}_N^1,\ldots,\hat{\lambda}_N^\Xi$, $\bar{r}_{T}^s$ is a sample Fr\'echet mean of the treated units in stratum $s$ for $s=1,\ldots,\Xi$, and the corresponding terms for the control group are defined analogously. The infimum is taken over all of the relevant weighted sample Fr\'echet mean sets; note a slight abuse of notation, as $\bar{r}_C(\cdot)$ and $\bar{r}_T(\cdot)$ are not necessarily functions here. This estimator is attractive because it reduces to $\textsc{aate}$ in the real case and is equivalent to $T_{2}(Z_N,R_N,S_N)$ in a completely randomized experiment, i.e., when $\Xi=1$. However, unlike $T_{2}(Z_N,R_N,S_N)$, this estimator is not necessarily consistent when $\Xi\neq 1$; see the following example.
    \end{remark}

    \begin{example}
            Let $M=S^2=\{(x,y,z):x^2+y^2+z^2=1\}$, $p$ be the north pole $(0,0,1)$ and $c$ any value in $(0,\pi/2)$. Suppose 
    \begin{align*}
    &\text{pr}\Big[(r_{T0},r_{C0},s_0)=\big((\sin c,0,\cos c),(-\sin c,0,\cos c),1\big)\Big] \\
    =&\text{pr}\Bigg[(r_{T0},r_{C0},s_0)=\Bigg(\bigg(-\frac{\sin c}{2},\frac{2\sin c}{\surd 3},\cos c\bigg),\bigg(\frac{\sin c}{2},\frac{2\sin c}{\surd 3},\cos c\bigg),2\Bigg)\Bigg] \\
    =&\text{pr}\Bigg[(r_{T0},r_{C0},s_0)=\Bigg(\bigg(-\frac{\sin c}{2},-\frac{2\sin c}{\surd 3},\cos c\bigg),\bigg(\frac{\sin c}{2},-\frac{2\sin c}{\surd 3},\cos c\bigg),2\Bigg)\Bigg] \\
    =&\frac{1}{3};
    \end{align*}
    thus, $d(p,r_{T0})=d(p,r_{C0})=c$ with probability 1. By the strong consistency of the sample Fr\'echet means, all of which will always be unique for this distribution, $\bar{r}_{T}^1$ converges to $(\sin c,0,\cos c)$, the unique Fr\'echet mean of $r_{T0}$ given $s_0=1$, and $\bar{r}_{T}^2$ to the geodesic midpoint $(-\sin t,0,\cos t)$ for some $t\in(0,\pi/2)$, the unique Fr\'echet mean of $r_{T0}$ given $s_0=2$. We know that $(\sin c,0,\cos c)$, $(-\sin t,0,\cos t)$ and $p$ lie on a geodesic, and by the positive curvature of $S^2$, $t=d(p,(-\sin t,0,\cos t))>c/2$. Therefore, if $\hat{\lambda}_N^1\rightarrow \text{pr}(s_0=1)=1/3$ almost surely and $\hat{\lambda}_N^2\rightarrow \text{pr}(s_0=2)=2/3$ almost surely, then $\bar{r}(\bar{r}_{T}^1,\bar{r}_{T}^2;\hat{\lambda}_N^1,\hat{\lambda}_N^2)$ converges almost surely to $(-\sin(2t/3-c/3),0,\cos(2t/3-c/3))$. Similarly, $\bar{r}(\bar{r}_{C}^1,\bar{r}_{C}^2;\hat{\lambda}_N^1,\hat{\lambda}_N^2)\rightarrow(\sin(2t/3-c/3),0,\cos(2t/3-c/3))$ almost surely, 
    and thus, by the continuous mapping theorem, our estimator in (6) converges almost surely to
    \begin{align*}
    &d\Bigg(\bigg(-\sin\bigg(\frac{2t}{3}-\frac{c}{3}\bigg),0,\cos\bigg(\frac{2t}{3}-\frac{c}{3}\bigg)\bigg),\bigg(\sin\bigg(\frac{2t}{3}-\frac{c}{3}\bigg),0,\cos\bigg(\frac{2t}{3}-\frac{c}{3}\bigg)\bigg)\Bigg)\\
    =&\frac{4t}{3}-\frac{2c}{3}>0.
    \end{align*}
    However, the unique Fr\'echet means of $r_{T0}$ and of $r_{C0}$ are clearly $p$, and $\textsc{aate}$ are 0. An analogously defined estimator of $\textsc{amte}$ would face the same problem.
    \end{example}

    \subsection{Confidence intervals and hypothesis testing}
    
    Define a statistical functional $H_{\alpha}$ by
    \begin{align*}
        &H_{\alpha}(F) \\
        =&\inf d\big(\arg\min_{p\in M} E_F\big[E_F\big[d(p,r_1)^\alpha\mid z_1=1,s_1\big]\big],\arg\min_{p\in M} E_F\big[E_F\big[d(p,r_1)^\alpha\mid z_1=0,s_1\big]\big]\big),
    \end{align*}
    where $E_F$ denotes the expectation when $(z_1,r_1,s_1)$ follows a distribution $F$, and the infimum is taken over all elements in the two $\arg\min$ sets. Let $\hat{F}_N$ be the empirical distribution of $(z_1,r_1,s_1),\ldots,(z_N,r_N,s_N)$.

    \begin{proposition} \label{bootci}
        All terms are as in Theorem \ref{l2sc}. In addition, assume that $(r_{Ti},r_{Ci},s_i)$ are independent and identically distributed samples from the distribution of $(r_{T0},r_{C0},s_0)$. For $\alpha=1,2$, $H_\alpha(F)=d(\mu_{\alpha C},\mu_{\alpha T})$ and $H_{\alpha}(\hat{F}_N)=T_{\alpha}(Z_N,R_N,S_N)$
    \end{proposition}

    A proof is provided in Section B.4 of the Supplementary Materials. Take a parameter $\theta=G(F)=\theta$ of $F$, where $G$ is some statistical functional, and an estimator $\hat{\theta}_N=G(\hat{F}_N)$. Recall that the $(1-\delta)$ bootstrap pivotal interval for $\theta$ is obtained from $B$ bootstrap replications of $\hat{\theta}_N$ calculated by resampling from $\hat{F}_N$. Thus, if the conditions in the above proposition are satisfied, we can obtain confidence intervals using $T_{\alpha}$ for $\textsc{aate}$ when $\alpha=2$ and $\textsc{amte}$ when $\alpha=1$. This may be the case in observational studies free of hidden bias with categorical covariates as each combination of covariates can serve as a stratum, but when the covariates are non-categorical, exact stratification/matching into finitely many strata is usually not possible. However, exact stratification/matching can be approximated by various matching methods. In this case, we may obtain confidence intervals by calculating $T_{\alpha}$ after performing matching and repeating the entire process, including matching, for each of the $B$ bootstrapped resamples. We will examine the empirical performance of the bootstrap pivotal interval in Section \ref{simexp}.

    These confidence intervals can be used to test the null hypothesis $\textsc{aate}=0$ or $\textsc{amte}=0$. 
    The proposed estimators for $\textsc{aate}$ and $\textsc{amte}$ are also reasonable as test statistics for Fisher's sharp null hypothesis, which states $r_{Ti}=r_{Ci}$ for all $i=1,\ldots,N$, because these estimators are expected to be small under $H_0$, allowing us to reject $H_0$ for large values. See Section A.2 of the Supplementary Materials for more information on Fisher's sharp null hypothesis and the randomization procedure used to test it.
    
    	\section{Numerical experiments} \label{exp}

    \subsection{Simulation study} \label{simexp}
    
    We performed four types of simulations to evaluate the large-sample performance of the proposed estimators for the following properties: 
    \begin{itemize}
    \item[1.] Consistency of the estimators in a classical randomized experiment,
    \item[2.] Consistency of the bootstrapped confidence interval coverage in a classical randomized experiment,
    \item[3.] Consistency of the estimators in a matched observational study, and 
    \item[4.] Consistency of the bootstrapped confidence interval coverage in a matched observational study. 
    \end{itemize}
    The data were generated on the two-sphere $S^2=\{(x,y,z):x^2+y^2+z^2=1\}$ for the randomized experiment simulations (types 1 and 2) and on two-dimensional hyperbolic space $\mathbb{H}^2=\{(x,y,z):-x^2+y^2+z^2=-1\}$ for the observational study simulations (types 3 and 4). These two spaces, equipped with appropriate metrics, are Riemannian manifolds of constant sectional curvature 1 and -1, respectively, among the most commonly encountered non-Euclidean metric spaces. For more details on the spheres and hyperbolic spaces as Riemannian manifolds, including explicit representations of the exponential, inverse exponential, and parallel transport maps, refer to Appendix B.1 and B.2 of \cite{Shin2022}.

    \subsubsection{Settings} \label{simexp2}
    Recall the notion of parallel transport on a Riemannian manifold $M$: intuitively, given a differentiable curve $c:[a,b]\rightarrow M$, not necessarily a geodesic, parallel transport allows one to translate a vector in $T_{c(a)}M$ to $T_{c(b)}M$. Assuming a unique geodesic exists between $p_1,p_2\in M$, we will denote parallel transport of $v\in T_{p_1}M$ to $T_{p_2}M$ along this unique geodesic $\Gamma_{p_1\rightarrow p_2}(v)$.

    In all four types of simulation, we generated the potential outcomes $(r_{Ti},r_{Ci})$ independently for each of $N$ units in the following way. Let $p=(1,0,0)^\top$, $v_1=(0,\pi/4,0)^\top$, $v_2=(0,0,-\pi/6)^\top$, $\zeta_T=\exp_p((0,1,0)^\top)$, and $\zeta_C=\exp_p((0,-1,0)^\top)$. We generated two independent covariates, $x_i^1$ and $x_i^2$, (here the superscripts are indices, not exponents) from the uniform $(-1/2,1/2)$ distribution, and defined $r_{Ti}'=\exp_{\zeta_T}(\Gamma_{p\rightarrow \zeta_T}(x^1v_1+x^2v_2))$ and $r_{Ci}'=\exp_{\zeta_C}(\Gamma_{p\rightarrow \zeta_C}(x^1v_1+x^2v_2))$. The normal distribution on a connected Riemannian manifold $M$ is defined in \cite{Fletcher2013} as $f(y;\mu,\sigma^2)=(1/C_M(\mu,\sigma^2)\exp(-d(y,\mu)^2/2\sigma^2)$, where $C_M(\mu,\sigma^2)=\int_M\exp(-d(y,\mu)^2/2\sigma^2)dy$. For details on generating data from the Riemannian normal distribution on spheres and hyperbolic spaces, including the calculation of the normalizing constant $C_M(\mu,\sigma^2)$, see Appendix A.3 of \cite{Shin2022}. We also generated a data point $y_i$ from the Riemannian normal distribution centred at $p$ for $\sigma^2=(\pi/8)^2$ and lifted it into $T_pM$, where $M=S^2$ or $\mathbb{H}^2$, via the inverse exponential map $\log_p$. Then, we defined $r_{Ti}=\exp_{r_{Ti}'}(\Gamma_{p\rightarrow r_{Ti}'}(\log_p(y_i)))$ and $r_{Ci}=\exp_{r_{Ci}'}(\Gamma_{p\rightarrow r_{Ci}'}(\log_p(y_i)))$. By symmetry, $\mu_{2T}=\mu_{1T}=\zeta_T$ and $\mu_{2C}=\mu_{1C}=\zeta_C$, so $\textsc{aate}=\textsc{amte}=2$. We performed additional experiments of types 1 and 3 by varying $\sigma^2$. These results are provided in Tables \ref{sreconsistency} and \ref{obsstudyconsistency}.

    In the observational study simulations, we performed full matching on the generated covariates $x_i^1$ and $x_i^2$ using the rank-based Mahalanobis distance with a propensity score caliper, while in the randomized experiment simulations, we assigned a unit to one of two strata based on $s_i=2-I\{x_i^1\geq0\}$; that is, if $x_i^1\geq0$, the unit is in stratum 1; otherwise, it is in stratum 2.

    Treatment $z_i$ was assigned as follows. Let $m_N^s$ be the number of units in stratum $s$. In the randomized experiment setting, $\lfloor (m_N^s+1)/2\rfloor$ units were randomly assigned to the treatment group, and $\lfloor m_N^s/2\rfloor$ to the control group. In the matched observational study, $z_i$ was generated independently for each unit based on $\text{pr}(z_i=1)=1/(1+\exp(-x_i^1-x_i^2))$. 

    We computed the weighted sample Fr\'echet means and geometric medians using a gradient descent algorithm based on the \texttt{GeodRegr} R package. As per Remark \ref{unique}, we treated the point that this algorithm converged to as the unique element of the appropriate $L_\alpha$ estimator set.

    When calculating $T_{2}$ and $T_{1}$ using the $N$ generated data points, in the stratified randomized experiment simulations, we took $\text{pr}(s_i=1)$ and $\text{pr}(s_i=2)$ as known quantities and fixed the stratum-wise weights $\hat{\lambda}_N^1$ and $\hat{\lambda}_N^2$ to the known values of $\text{pr}(s_i=1)=1/2$ and $\text{pr}(s_i=2)=1/2$, respectively, while, in the observational simulations, we set $\hat{\lambda}_N^s$ to be $m_N^s/N$ for each $s$. Then, in the simulations of types 2 and 4, we calculated confidence intervals using $B=500$ bootstrapped resamples.

    We repeated the above process 500 times for different values of $N$. For a given $N$, let $Z_{N l}$, $R_{N l}$, and $S_{N l}$ be the values of $Z_N$, $R_N$, and $S_N$, respectively, in the $l$th iteration of the experiment in the simulation types 1 and 3, and $C_{2 N l}$ and $C_{1 N l}$ be the bootstrapped confidence intervals for $\textsc{aate}$ and $\textsc{amte}$, respectively, in the $l$th iteration of the experiment in the simulation types 2 and 4. For each $N\in\{32,64,128,256,512,1024\}$ and $\alpha\in\{1,2\}$, we estimated mean absolute errors $(1/500)\sum_{l=1}^{500}\lvert T_{\alpha}(Z_{N l}, R_{N l}, S_{N l})-2\rvert$ for the simulations of types 1 and 3 and confidence interval coverages $(1/500)\sum_{l=1}^{500}I\{2\in C_{\alpha N l}\}$ for the simulations of types 2 and 4.

    \subsubsection{Results}
    \begin{table}[!ht]
    \centering
    \captionsetup{justification=justified}
    \caption{Estimates for mean absolute errors in experiments of types 1 and 3 and for confidence interval coverages in types 2 and 4, with estimated standard errors in parentheses.}
      \resizebox{\textwidth}{!}{%
    \begin{tabular}{cccccccc}
    \hline
    \multicolumn{1}{c}{Experiment type} & \multicolumn{1}{c}{Estimator} & $N=32$ & $N=64$ & $N=128$ & $N=256$ & $N=512$ & $N=1024$  \\
    \hline
    \multirow{4}{*}{1} &
    \multirow{2}{*}{$T_{2}$} & 0.122 & 0.080 & 0.060  & 0.040 & 0.030 & 0.021 \\ &  & (0.088) & (0.063) & (0.045) & (0.031) & (0.024) &(0.016) \\
    & \multirow{2}{*}{$T_{1}$} & 0.148 & 0.092 & 0.067 & 0.046 & 0.035 & 0.025 \\  & &  (0.104) & (0.072) & (0.053) & (0.035) & (0.027) & (0.019) \\
     \hline
    \multirow{4}{*}{2} &
    \multirow{2}{*}{$T_{2}$} & 0.924 & 0.932 & 0.930  & 0.930 & 0.940 & 0.942 \\ &  & (0.012) & (0.011) & (0.011) & (0.011) & (0.011) &(0.010) \\
    & \multirow{2}{*}{$T_{1}$} & 0.910 & 0.908 & 0.922 & 0.938 & 0.940 & 0.944 \\  & &  (0.013) & (0.013) & (0.012) & (0.011) & (0.011) & (0.010) \\
     \hline
    \multirow{4}{*}{3} &
    \multirow{2}{*}{$T_{2}$} & 0.134 & 0.092 & 0.061  & 0.043 & 0.031 & 0.022 \\ &  & (0.088) & (0.063) & (0.045) & (0.031) & (0.024) &(0.016) \\
    & \multirow{2}{*}{$T_{1}$} & 0.161 & 0.104 & 0.071 & 0.049 & 0.036 & 0.025 \\  & &  (0.104) & (0.072) & (0.053) & (0.035) & (0.027) & (0.019) \\
     \hline
    \multirow{4}{*}{4} &
    \multirow{2}{*}{$T_{2}$} & 0.872 & 0.896 & 0.898  & 0.900 & 0.908 & 0.914 \\ &  & (0.015) & (0.014) & (0.014) & (0.013) & (0.013) &(0.013) \\
    & \multirow{2}{*}{$T_{1}$} & 0.848 & 0.884 & 0.900 & 0.904 & 0.910 & 0.920 \\  & &  (0.016) & (0.014) & (0.013) & (0.013) & (0.013) & (0.012) \\
     \hline
    \end{tabular}
    }
    \label{simulres}
    \end{table}

    As listed in Tables \ref{simulres}-\ref{obsstudyconsistency}, the mean absolute error and coverage generally improve with $N$ in both the classical randomized experiment and the matched observational study settings. This is empirical evidence for the strong consistency of the estimators and for the reasonable asymptotic coverage of the confidence intervals in both settings. In both cases, the coverage is fairly close to 95\%, but we find that the coverage of the bootstrap pivotal interval tends to be lower than the stated confidence level. However, the coverage in the observational study case is somewhat lower than in the randomized experiment case. This may be influenced by the fact that the covariates were not categorical, so exact matching could only be approximated.

  \begin{table}[!ht]
    \centering
    \captionsetup{justification=justified}
    \caption{Estimates for mean absolute errors in randomized experiment simulations, with estimated standard errors in parentheses.}
    \begin{tabular}{cccccccc}
    \hline
    \multicolumn{1}{c}{$\sigma^2$} & \multicolumn{1}{c}{Estimator} & $N=32$ & $N=64$ & $N=128$ & $N=256$ & $N=512$ & $N=1024$  \\
    \hline
    \multirow{4}{*}{$(\pi/16)^2$}  & \multirow{2}{*}{$T_{2}$} & 0.066 & 0.044 & 0.032  & 0.022 & 0.016 & 0.011 \\ & &  (0.049) & (0.035) & (0.025) & (0.017) & (0.013) &(0.009) \\
     & \multirow{2}{*}{$T_{1}$} & 0.119 & 0.060 & 0.039 & 0.026 & 0.020 & 0.014 \\ & &  (0.197) & (0.099) & (0.031) & (0.021) & (0.016) & (0.011) \\
     \hline
     \multirow{4}{*}{$(\pi/4)^2$}  & \multirow{2}{*}{$T_{2}$} & 0.250 & 0.173 & 0.126  & 0.089 & 0.066 & 0.049 \\ & &  (0.182) & (0.137) & (0.094) & (0.083) & (0.063) &(0.035) \\
     & \multirow{2}{*}{$T_{1}$} & 0.292 & 0.189 & 0.139 & 0.094 & 0.073 & 0.052 \\ & &  (0.205) & (0.139) & (0.111) & (0.072) & (0.056) & (0.040) \\
    \hline
    \end{tabular}
    \label{sreconsistency}
    \end{table}
    
    \begin{table}[!ht]
    \centering
    \captionsetup{justification=justified}
    \caption{Estimates for mean absolute error in observational study simulations, with estimated standard errors in parentheses.}
    \begin{tabular}{cccccccc}
    \hline
    \multicolumn{1}{c}{$\sigma^2$} & \multicolumn{1}{c}{Estimator} & $N=32$ & $N=64$ & $N=128$ & $N=256$ & $N=512$ & $N=1024$ \\
    \hline
    \multirow{4}{*}{$(\pi/16)^2$}  & \multirow{2}{*}{$T_{2}$} & 0.072 & 0.048 & 0.032  & 0.022 & 0.016 & 0.011 \\ &  & (0.053) & (0.036) & (0.024) & (0.016) & (0.012) &(0.008) \\
     & \multirow{2}{*}{$T_{1}$} & 0.089 & 0.056 & 0.039 & 0.028 & 0.020 & 0.014 \\ & &  (0.074) & (0.044) & (0.030) & (0.021) & (0.015) & (0.010) \\
     \hline
     \multirow{4}{*}{$(\pi/4)^2$}  & \multirow{2}{*}{$T_{2}$} & 0.251 & 0.172 & 0.113  & 0.079 & 0.057 & 0.041 \\ &  & (0.185) & (0.124) & (0.087) & (0.059) & (0.045) &(0.030) \\
     & \multirow{2}{*}{$T_{1}$} & 0.305 & 0.192 & 0.130 & 0.089 & 0.066 & 0.045 \\ & &  (0.235) & (0.143) & (0.096) & (0.068) & (0.049) & (0.036) \\
    \hline
    \end{tabular}
    \label{obsstudyconsistency}
    \end{table}
    
       \subsection{Real data analysis}
     
     We investigate the causal relationship between Alzheimer's disease and the shape of the corpus callosum, a large white matter structure in the brain that facilitates communication between the two cerebral hemispheres. Our data, containing 186 units with Alzheimer's disease and 223 units without, are from the Alzheimer's Disease Neuroimaging Initiative study. Covariates include gender, handedness, marital status, years of education, retirement status, and age. The treatment variable $z$ is the diagnosis: $z=1$ if the unit has the disease, and $z=0$ otherwise. The outcome variable $r$, the planar shape of the corpus callosum, lies on the $(2K-4)$-dimensional Riemannian manifold $\Sigma_2^K$, which is Kendall's two-dimensional shape space for $K=50$ landmarks. It is the set of all equivalence classes of $K$-gons in the two-dimensional plane, where two $K$-gons are equivalent if one can be transformed into the other through translation, scaling, and rotation. For more details on $\Sigma_2^K$, refer to Appendix B.3 of \cite{Shin2022}, Section 3.11 of the online supplement to \cite{Cornea2017}, or Section 5.2.1 of \cite{Fletcher2013}. The corpus callosum shape data were obtained by extracting the planar shapes from the mid-sagittal slices of magnetic resonance images and segmenting them using the \texttt{FreeSurfer} and \texttt{CCSeg} packages, resulting in a $K\times2$ matrix. Each row of this matrix represents the planar coordinates of one of the $K=50$ boundary points used to define the shape of the corpus callosum, and the corresponding rows for different units represent the corresponding points on the boundaries of the shapes. Our data are from the Alzheimer's Disease Neuroimaging Initiative study and are provided by \cite{Cornea2017} at \url{http://www.bios.unc.edu/research/bias/software.html}. As with the simulations, we treated the point of convergence of a gradient descent algorithm as the unique element of the $L_\alpha$ estimator set in light of Remark \ref{unique}.

     \subsubsection{Matching}
     We implemented matching to approximate within-stratum conditional independence of $z_i$ and $(r_{Ti},r_{Ci})$. Gender, handedness, marital status, years of education, retirement status, and age were used as observed confounders. We considered a variety of matching methods: pair matching, full matching, full matching with restrictions on the maximum number of treatment or control units in a matched set, and almost exact pair and full matching with penalties for imbalances in age, which is thought to have an important effect on the shape of the corpus callosum. Full matching using the rank-based Mahalanobis distance with a propensity score caliper, which showed good covariate balance between the treated and control units under the assessment based on standardized differences, was used in the analysis.

     \subsubsection{Results}
    Our estimators were calculated to be $T_{2}=0.01819$ and $T_{1}=0.01775$. For researchers unfamiliar with this particular manifold, the interpretation of these numbers may be opaque. Figure \ref{fig:planarshapes} aids in this regard by illustrating the planar shapes of the weighted sample Fr\'echet  mean and geometric median corpus callosa of treated and control units. This then provides visual representations of the above numbers, which are respectively the geodesic distances on Kendall's two-dimensional shape space between the Fr\'echet mean planar shapes for the control and treatment groups and the geometric median planar shapes for those two groups.
    
    \begin{figure}[!ht]
       \centering
        \subcaptionbox{Sample Fr\'echet mean corpus callosa. \label{fig:l2cc}}[0.47\linewidth]{
        \includegraphics[width=\linewidth]{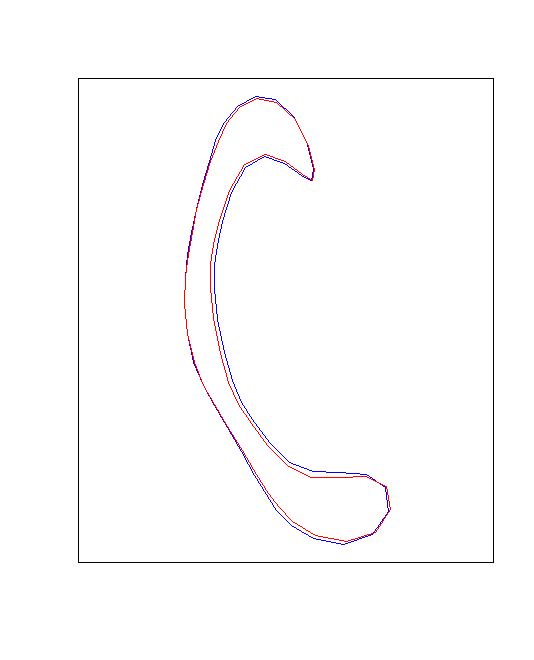}
        } 
        \subcaptionbox{Sample geometric median corpus callosa. \label{fig:l1cc}}[0.47\linewidth]{
        \includegraphics[width=\linewidth]{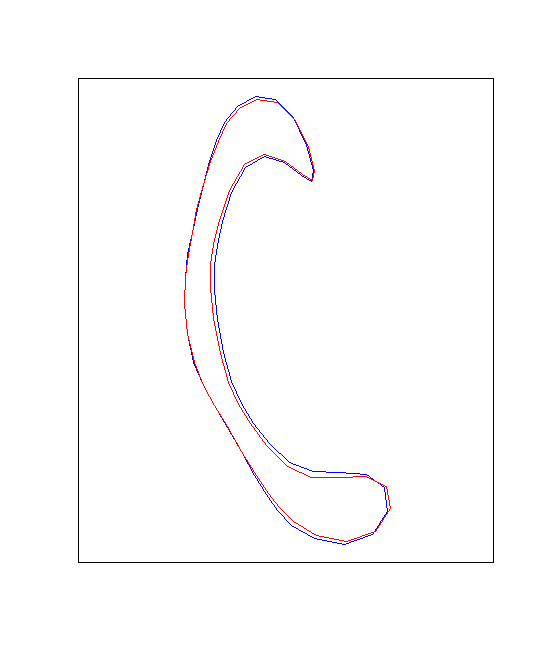}
        }
        \caption{The planar shapes of the sample Fr\'echet mean corpus callosa and the sample geometric median corpus callosa of units with (red) and without (blue) Alzheimer's disease.}
        \label{fig:planarshapes}
    \end{figure}
    
    With these test statistics, we performed randomization inference to test Fisher's null hypothesis by reassigning the treatment indicator within the matched sets. After performing the randomization process 1000 times, the $p$-values for both $T_{2}$ and $T_{1}$ were 0 and 0.002, respectively.

    In addition to our point estimates $T_{2}=0.01819$ and $T_{1}=0.01775$, bootstrapped 95\% pivotal confidence intervals for $\textsc{aate}$ and $\textsc{amte}$ were obtained as $(0.00915, 0.02252)$ and $(0.00633, 0.02146)$, respectively, where the number of bootstrap samples is 1000, and thus, we can reject the null hypotheses $H_0:\textsc{aate}=0$ and $H_0:\textsc{amte}=0$.

This data set exemplifies the usefulness of our metric space-based methodology because the shape is a non-Euclidean characteristic. However, suppose one wants to employ more conventional Euclidean-based methods for causal inference. One immediately runs into the problem of representing each shape as a point in Euclidean space; although each data point is represented as a $K\times 2$ matrix, it would not be appropriate to use this representation because it could treat two data points whose underlying planar configurations differ only by translation, scaling, and rotation as unequal. Therefore, the effects of translation, scaling, and rotation need to be removed as described in Appendix B.3 of \cite{Shin2022}, meaning each point is sent to Kendall's two-dimensional shape space anyway, even when using a Euclidean approach. If one still insists on treating these transformed $K\times 2$ matrices as points in $2K$-dimensional Euclidean space instead of the more natural $(2K-4)$-dimensional space $\Sigma_2^K$, the resulting point estimates are $T_{2}=0.02241$ and $T_{1}=0.02176$, with $p$-values for Fisher's sharp null hypothesis of 0.226 and 0.261, respectively. These are more than 1000 times larger than the equivalent tests on $\Sigma_2^K$, hinting that the power of these tests may be diminished by orders of magnitude when the inherent geometric structure of the data is not exploited. Using 1000 bootstrap samples, the bootstrapped 95\% pivotal confidence intervals for $\textsc{aate}$ and $\textsc{amte}$ under the Euclidean framework were respectively $(0.00204, 0.02922)$ and $(-0.00261, 0.02954)$, or $[0,0.02854)$ since $\textsc{amte}$ is non-negative. Testing the null hypotheses $H_0:\textsc{aate}=0$ and $H_0:\textsc{amte}=0$ with these confidence intervals, although the difference is not quite so stark as with the sharp null hypothesis, these intervals are significantly wider than their counterparts from $\Sigma_2^K$, suggesting lower power. The second interval even includes 0, and so we cannot reject $H_0:\textsc{amte}=0$ at the 95\% level, even though we could do so on $\Sigma_2^K$. These findings indicate that the use of existing Euclidean methods for causal inference may lead to reductions in power when the data set has a natural non-Euclidean structure; further research is required.
    
        
     \section{Conclusion} \label{conclusion}
        In this study, we defined the notions of absolute average and median treatment effects on metric spaces, including Riemannian manifolds. We proposed estimators for these quantities using stratification, as well as testing procedures and confidence intervals and proved the strong consistency of the estimators when the space is proper. Simulation experiments provided empirical evidence that these estimators and confidence intervals perform well. These estimators were used as test statistics for randomization testing of the sharp null hypothesis that there is no causal relationship between Alzheimer’s disease and the shape of the corpus callosum. The sharp null hypothesis was rejected using both test statistics, implying that Alzheimer's disease has a causal effect on the shape of the corpus callosum. We also used confidence intervals to reject the weak null hypotheses that $\textsc{aate}$ and $\textsc{amte}$ are 0. On the other hand, conventional Euclidean methods struggled to reject these hypotheses, possibly suggesting significantly worse power.
        
        There are many possible avenues for future research into causal inference on metric spaces. For example, the asymptotic distribution of the proposed estimators would be useful for calculating confidence intervals without bootstrapping. Furthermore, beyond $\textsc{aate}$ and $\textsc{amte}$, which measure the magnitude of the overall causal effect but provide no sense of direction, both the magnitude and direction of the overall causal effects can be investigated by defining a canonical sense of direction at every point using the boundary at infinity in the global non-positive curvature metric space, known as Hadamard space or complete \textsc{cat}(0) space. 
        
           \section*{Acknowledgement} 
This research was supported by the National Research Foundation of Korea (NRF) funded by the Korea government (2021R1A2C1091357).

\bibliographystyle{apalike}
\bibliography{refs}

\newpage

	\section*{Supplementary materials: Absolute average and median treatment effects as causal estimands on metric spaces}	

\appendix
\renewcommand{\theequation}{\thesection.\arabic{equation}}
\setcounter{equation}{0}
\section{Miscellaneous}
    
\subsection{Geodesic regression}\label{geodesic}

\subsubsection{Geodesic regression on Riemannian manifolds} \label{manifoldgeodesic}

       A topological manifold $M$ is a Hausdorff, second-countable topological space that locally resembles Euclidean space. A smooth manifold is a topological manifold with a so-called smooth structure over which calculus can be performed, and a Riemannian manifold $(M,g)$ is a smooth manifold $M$ equipped with a Riemannian metric $g$, that is, a smooth map that assigns an inner product to each point $p\in M$. We will refer to $(M,g)$ as $M$, suppressing mention of $g$ unless necessary. This $g$ can be used to define the lengths of piecewise continuously differentiable curves, and the distance between two points $p_1,p_2\in M$ connected by such a curve, is the infimum of the set of lengths of all paths with endpoints $p_1$ and $p_2$. Because path-connectedness and connectedness are equivalent on manifolds if $M$ is connected, one can then define a function $d:M\times M\rightarrow\mathbb{R}$ that satisfies the axiomatic properties of a metric by letting $d(p_1,p_2)$ be the distance between $p_1$ and $p_2$. The Hopf-Rinow theorem shows that if $M$ is a complete and connected Riemannian manifold, it is proper. 

       A geodesic $\gamma:I\rightarrow M$ on a Riemannian manifold, where $I\subset\mathbb{R}$ is a real interval and $0\in I$, can be defined by an initial point $\gamma(0)$ and $\dot{\gamma}(0)\in T_{\gamma(0)}M$, where $T_pM$ denotes the tangent space of $M$ at the point $p\in M$. Then the exponential map $\exp_p:U_p\rightarrow M$ is defined by $\exp_p(v)=\gamma_{p,v}(1)$, where $\gamma_{p,v}:I\rightarrow M$ is a geodesic for which $\gamma_{p,v}(0)=p$, $\dot{\gamma}_{p,v}(0)=v$; so geodesics on $M$ are of the form $\exp_p(tv)$. Here $U_p\subset T_pM$ is the largest neighborhood of $0\in T_pM$ for which the above makes sense; that is, for which $\gamma_{p,v}(t)$ exists in $M$ for all $t\in[0,1]$. If $M=\mathbb{R}^n$, $U_p=T_p\mathbb{R}^n\cong\mathbb{R}^n$ and $\exp_p(v)=p+v$. Intuitively, $\exp_p(v)$ is the point that results from wrapping a vector $v\in T_pM$ onto $M$. The Hopf-Rinow theorem shows that if $M$ is a complete connected Riemannian manifold, $U_p=T_pM$ in the definition of $\exp_p$ for all $p\in M$, meaning that geodesics starting from $p$ continue indefinitely in every direction; this property is called geodesic completeness. The inverse exponential map $\log_p:W_p\rightarrow T_pM$ is defined by $\log_p(\exp_p(v))=v$. Since $\exp_p$ is locally diffeomorphic, hence bijective, from some neighborhood of $0\in U_p$ to some neighborhood of $p$, we define $W_p\subset M$ to be the largest such neighborhood of $p$. If $M=\mathbb{R}^n$, $W_p=\mathbb{R}^n$ and $\log_p(q)=q-p$. For $q\in W_p$, $d(p,q)=\lVert \log_p(q)\rVert$.

       Let $M$ be a complete connected Riemannian manifold. In the simple geodesic regression model defined by \cite{Fletcher2013}, a response $y\in M$ given a covariate $x\in\mathbb{R}$ is generated from the following model
     \begin{equation*}
        y=\exp_{\exp_p(xv)}(\epsilon),
     \end{equation*}
     where $p\in M$, $v\in T_pM$ and $\epsilon\in T_{\exp_p(xv)}M$. For Euclidean space, $\exp_p(v)=p+v$; thus, the above geodesic model coincides with the multiple linear regression $y=p+xv+\epsilon$ when $M=\mathbb{R}^n$. 
     
     Given data points $(x_i,y_i)\in\mathbb{R}\times M$ for $i=1,\ldots,N$ and associated weights $w_i$ that satisfy $\sum_{i=1}^Nw_i=1$, we consider the weighted least squares, or $L_2$, estimator set by 
     \begin{equation*} 
        \arg\min_{(p,v):p\in M,v\in T_pM}\sum_{i=1}^Nw_id(\exp_p(x_iv),y_i)^2,
     \end{equation*}
    and the weighted least absolute deviations, or $L_1$, estimator set by
    \begin{equation*}
        \arg\min_{(p,v):p\in M,v\in T_pM}\sum_{i=1}^Nw_id(\exp_p(x_iv),y_i).
     \end{equation*}
     When $w_i=1/N$ for $i=1,\ldots,N$, we obtain the $L_2$ and $L_1$ estimators of \cite{Fletcher2013} and \cite{Shin2022}. These solutions can usually be found using gradient descent algorithms, such as the one provided in the \texttt{GeodRegr} R package.

     \subsubsection{Geodesic regression on geodesic spaces} \label{spacegeodesic}

In a metric space $M$, given an interval $I\subset \mathbb{R}$ containing 0, a geodesic is a map $\gamma:I\rightarrow M$ for which there exists some constant $u_\gamma\geq0$ and, for all $t\in I$, some neighborhood $J\subset I$ such that $t_1,t_2\in J$ implies $d(\gamma(t_1),\gamma(t_2))=u_\gamma\lvert t_2-t_1\rvert$; a minimal geodesic is a geodesic for which the aforementioned $J=I$. In $\mathbb{R}^n$, both geodesics and minimal geodesics are lines, line segments, and rays, while in $S^n$, geodesics are arcs of great circles but only those arcs of length less than or equal to $\pi$ are minimal geodesics. We call $u_\gamma$ the speed of the geodesic. If $0\in I$, we say that $\gamma$ starts at $\gamma(0)$. A geodesic space is a metric space for which any two points $p,q\in M$ can be joined by at least one minimal geodesic, i.e., there exists a minimal geodesic $\gamma$ such that $\gamma(0)=p$ and $\gamma(1)=q$. If two geodesics coincide on some real interval, then a geodesic can be defined on the interval that is the union of their domains; doing so for all such geodesics will produce what we call a geodesic of the maximal domain. In the case of connected Riemannian manifolds, the above definition of a geodesic coincides with the Riemannian one and $u_\gamma=\lVert\gamma'(0)\rVert$. The Hopf-Rinow theorem shows that if $M$ is a complete and connected Riemannian manifold, it is a geodesic space.
          
    Let $M$ be a geodesic space and  $D_p$ denote the set of geodesics $\gamma_p$ of the maximal domain starting at $p$. When $M$ is a complete connected Riemannian manifold, there is a natural bijection between the tangent spaces at $M$, and a $D_p$ mapping $v\in T_pM$ to the geodesic in $D_p$ defined by $\exp_p(tv)$ exists. Therefore, given a covariate $x\in\mathbb{R}$ and a response $y\in M$, one could easily generalize the simple geodesic regression model defined by \cite{Fletcher2013} to geodesic spaces by $y=\epsilon_{\gamma_{p}(x)}(1)$, where $\gamma_p\in D_p$ and $\epsilon_{\gamma_{p}(x)}\in D_{\gamma_{p}(x)}$. This model reduces to a weighted version of the geodesic regression model of \cite{Fletcher2013} when $M$ is a complete connected Riemannian manifold. 
    Then, given data points $(x_i,y_i)\in\mathbb{R}\times M$ for $i=1,\ldots,N$ and associated weights $w_i$ that satisfy $\sum_{i=1}^N w_i = 1$, we define the $L_\alpha$ estimator set by 
     \begin{equation} \label{l2estimator}
        \arg\min_{(p,\gamma_p):p\in M,\gamma_p\in D_p}\sum_{i=1}^Nw_id(\gamma_{p}(x_i),y_i)^\alpha.
     \end{equation}
    A minimizer may not exist or be unique.

     Considering only constant $\gamma_p$ by fixing the speed of the geodesic $u_{\gamma_p}$ at 0 in (\ref{l2estimator}) and optimizing with respect to $p$, we get the definition of the weighted sample $L_\alpha$ estimator set, $\arg\min_{p\in M}\sum_{i=1}^Nw_id(p,y_i)^\alpha$, which is equivalent to the set of minimizers of (2) for the random element whose distribution has mass $w_i$ at $y_i\in M$ for $i=1,\ldots,N$.

\subsection{Testing Fisher's sharp null hypothesis}\label{fisher}

        Let $r_{Ti}$ and $r_{Ci}$  ($i=1,\ldots,N$) be instances of treatment and control potential outcomes, respectively, which can be collected as ordered $N$-tuples $R_{TN}$ and $R_{CN}$, respectively. Fisher's sharp null hypothesis states
    \begin{equation*}
    H_0:R_{TN}=R_{CN},
    \end{equation*}  
      or $r_{Ti}=r_{Ci}$ for all $i=1,\ldots,N$. Associated with unit $i$ are a stratum $s_i$, a binary treatment variable $z_i$, and an observed outcome
    \begin{equation*}
    r_i=r_i(z_i)=\begin{cases} r_{Ti} &~~\mbox{if $z_i=1$} \\
    r_{Ci} &~~\mbox{if $z_i=0$.}
    \end{cases}
    \end{equation*}
    These values can also be collected into ordered $N$-tuples $S_N$, $Z_N$, and $R_N=R_N(Z_N)$, respectively. Consider $R_{TN}$, $R_{CN}$, and $S_N$ to be fixed so that the only source of randomness is $Z_N$. Then under $H_0$, $R_N$ is fixed, and $R_{TN}$ and $R_{CN}$ are known: $R_{TN}=R_{CN}=R_N$. So, in a randomized experiment, where the assignment mechanism, that is, the distribution of $Z_N$ given $R_{TN}$, $R_{CN}$ and $S_N$, is known, the exact distribution of any test statistic $T(Z_N,R_N,S_N)$ is characterized by 
    \begin{equation} \label{ri}
    \text{pr}\big(T(Z_N,R_N,S_N)=T(Z,R_N,S_N)\big)=\text{pr}(Z_N=Z \mid R_{TN},R_{CN},S_N).
    \end{equation}
    Since the number of possible values for $Z$ is finite, exact $p$-values can be obtained by calculating $\text{pr}(Z_N=Z \mid R_{TN},R_{CN},S_N)$ for all possible $Z$. This process is called randomization inference. This number of possible values is usually too large, and an approximation is used by generating a large number of samples from this known distribution.

    Typically, in a stratified randomized experiment, the assignment mechanism is unconfounded so that the assignment mechanism does not depend on $(R_{TN},R_{CN})$, and chosen so that the total number of treated units in each stratum, $m_{TN}^s=\sum_{k=1}^Nz_kI\{s_k=s\}$, is fixed and each possible treatment assignment has equal probability: $\text{pr}(Z_N=Z \mid R_{TN},R_{CN},S_N)=\text{pr}(Z_N=Z \mid S_N)=1/K$, where $K$ is the total number of possible treatment assignments.

    Let $M_{TN}=(m_{TN}^1,\ldots,m_{TN}^S)^\top$ be the vector consisting of the total numbers of treated units in each stratum. In an observational study that is free of hidden bias (that is, with no unmeasured confounder), it can be shown that exact stratification/matching on $x_i$ given $M_{TN}$ leads to $\text{pr}(Z_N=Z \mid M_{TN})=1/K$, where $K$ is the total number of possible treatment assignments given $M_{TN}$. That is, given $M_{TN}$, an observational study free of hidden bias and with exact stratification/matching mimics a stratified randomized experiment with the assignment mechanism of the type described in the previous paragraph and Fisher's sharp null hypothesis can be tested in the same way. In practice, exact stratification/matching is not possible for non-discrete covariates, and various matching methods are used to approximate it.


\section{Proofs}	\label{proofs}	
\setcounter{equation}{0}
\subsection{Proof of Theorem 1} \label{appenl2equiv}

Theorem 1 follows immediately from the following result and the fact that $u_{\bar{\gamma}_p}=\lVert\bar{\gamma}_p'(0)\rVert$ when $M$ is a complete connected Riemannian manifold.

\begin{theorem}
    Let $M$ be a geodesic space and $\alpha=1$ or $2$. Then $\inf_{(\bar{p},\bar{\gamma}_p)\in\bar{G}^\alpha}u_{\bar{\gamma}_p}=T_{\alpha}(Z_N,R_N,S_N)$, where $\bar{G}^\alpha$ is the weighted simple geodesic regression $L_\alpha$ estimator set from (\ref{l2estimator}) of the points $(z_i,r_i)\in\mathbb{R}\times M$, $i=1,\ldots,N$ with weights $W_{iN}$, and is invariant with respect to $\beta_T$ and $\beta_C$.
\end{theorem}

   \begin{proof}
Denote $\inf_{(\bar{p},\bar{\gamma}_p)\in\bar{G}^\alpha}u_{\bar{\gamma}_p}$ by $T_{\alpha, 2}(Z_N,R_N,S_N)$. For any $\alpha\in\{1,2\}$, $(\bar{p},\bar{\gamma}_p)\in\bar{G}^\alpha$, $\bar{r}_T\in\bar{A}_{TN}^\alpha$ and $\bar{r}_C\in\bar{A}_{CN}^\alpha$, $Q(\bar{p},\bar{\gamma}_p)$ defined as $\sum_{i=1}^NW_{iN}d(\bar{\gamma}_p(z_i),r_i)^\alpha$ satisfies
\begin{align} \label{ineq}
Q(\bar{p},\bar{\gamma}_p)&=\beta_T\sum_{i:z_i=1}w_{TiN}d(\bar{\gamma}_p(1),r_i)^\alpha+\beta_C\sum_{i:z_i=0}w_{CiN}d(\bar{p},r_i)^\alpha \nonumber \\
&=\beta_T\sum_{i=1}^Nw_{TiN}d(\bar{\gamma}_p(1),r_i)^\alpha+\beta_C\sum_{i=1}^Nw_{CiN}d(\bar{p},r_i)^\alpha \nonumber \\
&\geq\beta_T\sum_{i=1}^{N}w_{TiN}d(\bar{r}_T,r_i)^\alpha+\beta_C\sum_{i=1}^{N}w_{CiN}d(\bar{r}_C,r_i)^\alpha,
\end{align}
where the last inequality follows immediately from the definitions of $\bar{A}_{TN}^\alpha$ and $\bar{A}_{CN}^\alpha$. As $M$ is a geodesic space, there is at least one minimal geodesic from $\bar{r}_C$ to $\bar{r}_T$; that is, a geodesic $\bar{\gamma}$ such that $\bar{\gamma}(0)=\bar{r}_C$, $\bar{\gamma}(1)=\bar{r}_C$ and the speed $u_{\bar{\gamma}}$ is $d(\bar{r}_C,\bar{r}_T)$. $Q(\bar{\gamma}(0),\bar{\gamma})$ equals (\ref{ineq}), so $(\bar{\gamma}(0),\bar{\gamma})\in\bar{G}^\alpha$. So there is at least one element of $\bar{G}^\alpha$  for which $u_{\bar{\gamma}}=d(\bar{r}_C,\bar{r}_T)$; this implies that $\{d(\bar{r}_C,\bar{r}_T):\bar{r}_T\in\bar{A}_{TN}^\alpha,\bar{r}_C\in\bar{A}_{CN}^\alpha\}\subset\{u_{\bar{\gamma}_p}:(\bar{p},\bar{\gamma}_p)\in\bar{G}^\alpha\}$, so $T_{\alpha, 2}(Z_N,R_N,S_N)\leq T_{\alpha}(Z_N,R_N,S_N)$.

On the other hand, the equality of $Q(\bar{\gamma}(0),\bar{\gamma})$ and (\ref{ineq}) means that $(\bar{p},\bar{\gamma}_p)\in\bar{G}^\alpha$ if and only if $Q(\bar{p},\bar{\gamma}_p)$ equals (\ref{ineq}), or equivalently, $\bar{\gamma}_p(1)\in\bar{A}_{TN}^\alpha$ and $\bar{p}\in\bar{A}_{CN}^\alpha$. So $(\bar{p},\bar{\gamma}_p)\in\bar{G}^\alpha$ implies that 

\begin{equation*}
u_{\bar{\gamma}_p}\geq d(\bar{p},\bar{\gamma}_p(1))\geq \inf_{\bar{r}_T\in\bar{A}_{TN}^\alpha,\bar{r}_C\in\bar{A}_{CN}^\alpha}d(\bar{r}_C,\bar{r}_T)=T_{\alpha}(Z_N,R_N,S_N),
\end{equation*} 
where the first equality is achieved if $\bar{\gamma}$ is a minimal geodesic between $\bar{p}$ and $\exp_{\bar{p}}(\bar{v}))$. Taking the infimum over all $(\bar{p},\bar{\gamma}_p)\in\bar{G}$ in the above then gives $T_{\alpha, 2}(Z_N,R_N,S_N)\geq T_{\alpha}(Z_N,R_N,S_N)$.

The conclusions of the above two paragraphs imply $T_{\alpha, 2}(Z_N,R_N,S_N)=T_{\alpha}(Z_N,R_N,S_N)$ regardless of the values of $\beta_T$ and $\beta_C$.
    \end{proof}

  \subsection{Proof of Theorem 2} \label{appenl2gzation}



    The proof of Theorem 2 requires several lemmas. For notational simplicity, we will suppress the mention of $\omega\in\Omega$ for all quantities, points, maps, and sets apart from $N_1,N_2,N_3,N_4$ and $N_5$. We begin with some lemmas. Let all relevant terms be as defined in the statement of Theorem 2 or as most recently defined in this subsection unless otherwise stated.

    For any $s\in\{1,\ldots,\Xi\}$ and $j\in\mathbb{Z}^+$, define $\Omega_*^s=\{m_N^s\rightarrow\infty\}$ and
    \begin{equation*}
    k_j^s=\begin{cases}
        \arg\min_{i\in\mathbb{Z}^+}\{\sum_{i=1}^kI\{s_i=s\}=j\} &~~\mbox{on $\Omega_*^s$} \\
        1 &~~\mbox{on $(\Omega_*^s)^c$,}
    \end{cases}
    \end{equation*}
    so that on $\Omega_*^s$, $k_j^s$ is the $j$th positive integer for which $s_i=s$.
    
    \begin{lemma} \label{rv}
    The map $k_j^s$ is a random element defined on $(\Omega,\mathcal{F},\text{pr})$, as are the maps $y_{(j)}^s$, defined to be $y_{k_j^s}$, and $s_{(j)}^s$, defined to be $s_{k_j^s}$.
    \end{lemma}
    \begin{proof}
        For $k_j^s$, it suffices to show that $\{k_j^s=t\}$ is measurable for each $t\in\mathbb{Z}^+$. Noting that $\{k_j^s=t\}=\{\sum_{i=1}^{t-1}I\{s_i=s\}=j-1\}\bigcap\{s_t=s\}\bigcap\Omega_*^s\in\mathcal{F}$ for $t>1$ and $\{k_j^s=1\}=(\{j=1\}\bigcap\{s_1=s\})\bigcup(\Omega_*^s)^c\in\mathcal{F}$ completes the proof.

        Then for all $t\in\mathbb{Z}^+$, $y_{(j)}^s=Y_t$ on $\{k_j^s=t\}$, which is measurable since $\{k_j^s=t\}$ is measurable for all $t\in\mathbb{Z}^+$. Similarly, $s_{k_j^s}$ is measurable.
    \end{proof}
    
    \begin{lemma} \label{iid}
        The random elements $y_{(j)}^s \mid \Omega_*^s$ and $y_0 \mid (s_0=s)$ are identically distributed, while $y_{(1)}^s \mid \Omega_*^s,y_{(2)}^s \mid \Omega_*^s,\ldots$ are independent. 
    \end{lemma}
    \begin{proof}
    If $\text{pr}(\Omega_*)=1$ for some $\Omega_*\in\mathcal{F}$, $\text{pr}(W \mid \Omega_*)=\text{pr}(W)$ for any set $W\in\mathcal{F}$, and hence sets (and random elements) are independent with respect to $\text{pr}$ if and only if they are independent with respect to $\text{pr}(\cdot \mid \Omega_*)$.

For $t,t'\in\mathbb{Z}^+\bigcup\{0\}$ such that $t'\geq t+1$, let $E(t,t')$ be the random variable defined by
        \begin{equation*}
        E(t,t')=\begin{cases}
            0&~~\mbox{if $t'=t+1$} \\
            \sum_{i=t+1}^{t'-1}I\{s_i=s\}&~~\mbox{if $t'\geq t+2$.}
        \end{cases}
        \end{equation*}
Then for any $m\in\mathbb{Z}^+$, any $j_1,\ldots,j_m\in\mathbb{Z}^+$ such that $j_1<\cdots<j_m$, and any $t_1,\ldots,t_m\in\mathbb{Z}^+\bigcup\{0\}$ such that $t_1<\cdots<t_m$, $\{k_{j_1}=t_1,\ldots,k_{j_m}=t_m\}\bigcap\Omega_*^s=\{E(0,t_1)=j_1-1,s_{t_1}=s,E(t_1,t_2)=j_2-j_1-1,s_{t_2}=s,\ldots,E(t_{j_{m-1}},t_{j_m})=j_m-j_{m-1}-1,s_{t_m}=s\}\bigcap\Omega_*^s$, and $E(0,t_1),(Y_{t_1},s_{t_1}),E(t_1,t_2)$, $(Y_{t_2},s_{t_2}),\ldots,E(t_{j_{m-1}},t_{j_m}),(Y_{t_m},s_{t_m})$ are independent.

        Using the above observations and the assumptions that the $(y_i,s_i)$ are independent and the $y_i \mid (s_i=s)$ are distributed identically to $Y \mid (s_0=s)$, for any $Q$ in the Borel $\sigma$-algebra $\mathcal{B}$ of $M$,
        \begin{align*}
            \text{pr}(y_{(j)}^s\in Q \mid \Omega_*^s)&=\sum_{t=1}^\infty \text{pr}(y_t\in Q, k_j^s=t \mid \Omega_*^s) \\
            &=\sum_{t=1}^\infty \text{pr}(y_t\in Q,E(0,t)=j-1,s_t=s \mid \Omega_*^s) \\
            &=\sum_{t=1}^\infty \text{pr}(y_t\in Q,s_t=s \mid \Omega_*^s)\text{pr}(E(0,t)=j-1 \mid \Omega_*^s) \\
            &=\sum_{t=1}^\infty \text{pr}(y_t\in Q \mid s_t=s,\Omega_*^s)\text{pr}(s_t=s \mid \Omega_*^s)\text{pr}(E(0,t)=j-1 \mid \Omega_*^s) \\
            &=\sum_{t=1}^\infty \text{pr}(y_0\in Q \mid s_0=s)\text{pr}(E(0,t)=j-1,s_k=s \mid \Omega_*^s) \\
            &=\text{pr}(y_0\in Q \mid s_0=s)\sum_{t=1}^\infty \text{pr}(k_j^s=t \mid \Omega_*^s) \\
            &=\text{pr}(y_0\in Q \mid s_0=s),
        \end{align*}
        proving the first statement.
        
        Letting $t_0=0$ and using this result, the observations in the first two paragraphs of this proof, and the assumptions that the $(y_i,s_i)$ are independent and the $y_i \mid (s_i=s)$ are distributed identically to $y_0 \mid (s_0=s)$, for any $Q$ in the Borel $\sigma$-algebra $\mathcal{B}$ of $M$,
        \begin{align*}
            &\text{pr}\bigg(\bigcap\limits_{l=1}^m\{y_{(j_l)}^s\in Q_l\}\,\,\,\bigg|\,\,\,\Omega_*^s\bigg) \\
            =&\sum_{t_1<\cdots <t_m}\text{pr}\bigg(\bigcap\limits_{l=1}^m\{y_{t_l}\in Q_l,k_{j_l}=t_l\}\,\,\bigg|\,\,\Omega_*^s\bigg) \\
            =&\sum_{t_1<\cdots <t_m}\text{pr}\bigg(\bigcap\limits_{l=1}^m\{y_{t_l}\in Q_l,E(t_{l-1},t_l)=j_l-1,s_{t_l}=s\}\,\,\bigg|\,\,\Omega_*^s\bigg) \\
            =&\sum_{t_1<\cdots <t_m}\bigg(\prod_{l=1}^m\text{pr}(y_{t_l}\in Q_l,s_{t_l}=s \mid \Omega_*^s)\bigg)\text{pr}\bigg(\bigcap\limits_{l=1}^m\{E(t_{l-1},t_l)=j_l-1\}\,\,\bigg|\,\,\Omega_*^s\bigg) \\
            =&\sum_{t_1<\cdots <t_m}\bigg(\prod_{l=1}^m\text{pr}(y_{t_l}\in Q_l \mid s_{t_l}=s,\Omega_*^s)\text{pr}(s_{t_l}=s \mid \Omega_*^s)\bigg) \\
            &\qquad\text{pr}\bigg(\bigcap\limits_{l=1}^m\{E(t_{l-1},t_l)=j_l-1\}\,\,\bigg|\,\,\Omega_*^s\bigg) \\
            =&\sum_{t_1<\cdots <t_m} \bigg(\prod_{l=1}^m\text{pr}(y_{t_l}\in Q_l \mid s_{t_l}=s,\Omega_*^s)\bigg)\text{pr}\bigg(\bigcap\limits_{l=1}^m\{E(t_{l-1},t_l)=j_l-1,s_{t_l}=s\}\,\,\bigg|\,\,\Omega_*^s\bigg) \\
            =&\bigg(\prod_{l=1}^m\text{pr}(y_{t_l}\in Q_l \mid s_{t_l}=s,\Omega_*^s)\bigg)\sum_{t_1<\ldots<t_m}\text{pr}(k_{j_1}=t_1,\ldots,k_{j_m}=t_m \mid \Omega_*^s) \\
            =&\prod_{l=1}^m\text{pr}(y_0\in Q_l \mid s_0=s,\Omega_*^s) \\
            =&\prod_{l=1}^m\text{pr}(y_{(j_l)}\in Q_l \mid \Omega_*^s),
        \end{align*}
        completing the proof.
    \end{proof}
    
    \begin{lemma} \label{l2asc}
    Let $M$ be any metric space and $\alpha\in\{1/2,1,2\}$. If $f_{\alpha}(p)$, defined in (1) of the main paper, is finite,
    \begin{equation*}
    \hat{f}_{\alpha N}(p)=\sum_{i=1}^Nw_{iN}d(p,y_i)^\alpha=\sum_{s=1}^\Xi\hat{\lambda}_N^s\Bigg(\frac{1}{m_N^s}\sum_{i:s_i=s,1\leq i\leq N}d(p,y_i)^\alpha\Bigg),
    \end{equation*}
    converges almost surely to $f_{\alpha}(p)$.

    \end{lemma}
    \begin{proof} By Lemma \ref{iid}, $y_{(1)}^s \mid (m_N^s\rightarrow\infty),y_{(2)}^s \mid (m_N^s\rightarrow\infty),\ldots$ are independent and distributed identically to $y_0 \mid (s_0=s)$ for each $s=1,\ldots,\Xi$.
    
    So
    \begin{align*}
    &\text{pr}\bigg(\frac{1}{m_N^s}\sum_{i:s_i=s,1\leq i\leq N}d(p,y_i)^\alpha\rightarrow E(d(p,y_0)^\alpha \mid s_0=s)\,\,\bigg|\,\, m_N^s\rightarrow\infty\bigg)\\
    =&\text{pr}\bigg(\lim_{N'\rightarrow\infty}\frac{1}{N'}\sum_j^{N'}d(p,y_{(j)}^s)^\alpha=E(d(p,y_0)^\alpha \mid s_0=s)\,\,\bigg|\,\, m_N^s\rightarrow\infty\bigg)\\
    =&1
    \end{align*}
    by the strong law of large numbers. Then since $\text{pr}(m_N^s\rightarrow\infty)=1$,
    \begin{equation*}
    \text{pr}\bigg(\frac{1}{m_N^s}\sum_{i:s_i=s,1\leq i\leq N}d(p,y_i)^\alpha\rightarrow E(d(p,y_0)^\alpha \mid s_0=s)\bigg)=1.
    \end{equation*}
    This is true for all $s=1,\ldots,\Xi$, and therefore, since almost sure convergence is preserved by products and sums, and $\hat{\Lambda}_N\rightarrow(\text{pr}(s_0=1),\ldots,\text{pr}(s_0=\Xi))$ almost surely, $\hat{f}_{\alpha N}(p)$ converges almost surely to 
    \begin{align*}
    \sum_{s=1}^\Xi\text{pr}(s_0=s)E(d(p,y_0)^\alpha \mid s_0=s)=E(E(d(p,y_0)^\alpha \mid s_0))=E(d(p,y_0)^\alpha),
    \end{align*}
    which is the definition of $f_{\alpha}(p)$.
    \end{proof}

    \begin{lemma} \label{l2basic}
    Let $M$ be a proper metric space and $\alpha\in\{1,2\}$. If $f_{\alpha}(p^*)$ is finite for some $p^*\in M$, then $f_{\alpha}$ is finite and continuous on all of $M$ and the Fr\'echet mean set is nonempty and compact.
    \end{lemma}
    \begin{proof}
        The $\alpha=2$ case is proven in Theorem 2.1(a) of \cite{Bhattacharya2003}. Even though the theorem states that $M$ is a complete connected Riemannian manifold, the proof is valid for any proper metric space, as the authors note in Remark 2.3 of the same paper.

        The finiteness and continuity of $f_{1}$ follow immediately from the triangle inequality and integration. Letting $\xi=\inf_{p\in M}f_{1}(p)$, there exists by continuity of $G$ a sequence $p_n$ such that $f_{1}(p_n)$ goes to $\xi$. Then $\{f_{1}(p_n)\}$ is a bounded sequence. By the triangle inequality, $d(p_1,p_n)\leq d(p_1,X)+d(X,p_n)$, and taking expected values gives 
    \begin{equation} \label{bound}
    d(p_1,p_n)\leq f_{1}(p_1)+f_{1}(p_n).
    \end{equation}
    This and the boundedness of $\{f_{1}(p_n)\}$ imply that $\{p_n\}$ is also bounded. Proper metric spaces are complete, so the bounded sequence $\{p_n\}$ has a subsequence $\{p_{n'}\}$ that converges to some $p'$. $\{G(p_{n'})\}$ converges to $\xi$, so by the continuity of $f_{1}$, $f_{1}(p')=\xi$ and the geometric median set is nonempty. Since the geometric median set is the pre-image of the closed set $\{\xi\}$ under the continuous $f_{1}$, it is also closed. Replacing $p_1$ and $p_n$ in (\ref{bound}) with $p'$ and any minimizer $m\in g^{-1}(\xi)$, $d(p',m)\leq 2\xi$, so $g^{-1}(\xi)$ is bounded too. Therefore, it is compact as $M$ is proper.
    \end{proof}

    The rest of these lemmas and their proofs are heavily based on the proof of Theorem 2.3 in \cite{Bhattacharya2003}. 

    \begin{lemma} \label{l2sup}
    Let $M$ be a proper metric space  and $\alpha\in\{1,2\}$, and suppose $f_\alpha(p^*)$ is finite for some $p^*\in M$. For any compact $K\subset M$, 
    \begin{equation}
    \sup_{p\in K}\lvert\hat{f}_{\alpha N}(p)-f_{\alpha}(p)\rvert\rightarrow 0
    \end{equation}
    almost surely.
    \end{lemma}
    \begin{proof} 
We have a triangle inequality of sorts for the square root of $d$: for $p_1,p_2,p_3\in M$,
    \begin{align} \label{triangle}
    &d(p_1,p_2)^{1/2}+d(p_2,p_3)^{1/2} \nonumber \\
    =&(d(p_1,p_2)+d(p_2,p_3)+2(d(p_1,p_2)d(p_2,p_3))^{1/2})^{1/2}\geq d(p_1,p_3)^{1/2}
    \end{align}
    by the regular triangle inequality and the non-negativity of $d$.
    
    Define $\mathrm{diam}(K)=\sup_{p,p^*\in K}d(p,p^*)$, which is finite by the compactness, and hence boundedness, of $K$. By Lemma \ref{l2basic}, $f_{\alpha}(p)$ is finite for all $p\in M$, and so is $f_{\alpha/2}(p)$ since $f_{\alpha/2}(p)^2<f_{\alpha}(p)$ by Jensen's inequality. Now fixing any $p_0\in K$, we use the triangle inequality if $\alpha=2$ and (\ref{triangle}) if $\alpha=1$ to show that there exists some $N_2(\omega)<\infty$ for all $\omega\in\Omega_2$, defined as $\{\omega:\hat{f}_{\alpha/2,N}(p_0)\rightarrow f_{\alpha/2}(p_0)\}\in\mathcal{F}$, such that, for all $N\geq N_2(\omega)$,
    \begin{align*}
    \sup_{p\in K}\hat{f}_{\alpha/2,N}(p)&=\sup_{p\in K}\sum_{i=1}^Nw_{iN}d(p,y_i)^{\alpha/2} \\
    &\leq\sup_{p\in K}\sum_{i=1}^Nw_{iN}(d(p_0,y_i)^{\alpha/2}+d(p,p_0)^{\alpha/2}) \\
    &\leq\sum_{i=1}^Nw_{iN}(d(p_0,y_i)^{\alpha/2}+\sup_{p\in K}d(p,p_0)^{\alpha/2}) \\
    &\leq f_{\alpha/2}(p_0)+1+\mathrm{diam}(K)^{\alpha/2} \\
    &<\infty,
    \end{align*}
    where $\hat{f}_{\alpha/2,N}(p_0)$ is as defined in Lemma \ref{l2asc}. Define $A=f_{\alpha/2}(p_0)+1+\mathrm{diam}(K)^{\alpha/2}$ and, for a fixed $\epsilon_0>0$, $\delta_1=(\epsilon_0/6A)^{2/\alpha}$. Using the above result and the triangle inequality if $\alpha=2$ and (\ref{triangle}) if $\alpha=1$, we have
    \begin{align} \label{ppstar}
    &\sup_{p,p*\in K:d(p,p^*)<\delta_1}\lvert\hat{f}_{\alpha N}(p)-\hat{f}_{\alpha N}(p^*)\rvert \nonumber\\
    =&\sup_{p,p^*\in K:d(p,p^*)<\delta_1}\bigg\lvert\sum_{i=1}^Nw_{iN}(d(p,y_i)^\alpha-d(p^*,y_i)^\alpha)\bigg\rvert \nonumber \\
    \leq&\sup_{p,p*\in K:d(p,p^*)<\delta_1}\sum_{i=1}^Nw_{iN}\big(\lvert(d(p,y_i)^{\alpha/2}-d(p^*,y_i))^{\alpha/2}\rvert\big)\big((d(p,y_i)^{\alpha/2}+d(p^*,y_i)^{\alpha/2})\big) \nonumber \\
    \leq&\sup_{p,p*\in K:d(p,p^*)<\delta_1}\sum_{i=1}^Nw_{iN}\big(d(p,p^*)^{\alpha/2}\big)\big(2d(p,y_i)^{\alpha/2}\big) \nonumber \\
    <&2\delta_1^{\alpha/2}\sup_{p\in K}\sum_{i=1}^Nw_{iN}d(p,y_i)^{\alpha/2} \nonumber \\
    =&\frac{\epsilon_0}{3}
    \end{align}
    for all $\omega\in\Omega_2$, $N\geq N_2(\omega)$. Now by the compactness of $K$ and the continuity of $f_{\alpha}$, proven in Lemma \ref{l2basic}, $f_{\alpha}$ is uniformly continuous on $K$ and thus there exists some $\delta_2$ such that 
    \begin{equation} \label{ppstar2}
    \sup_{p,p^*\in K:d(p,p^*)<\delta_2}\lvert f_{\alpha}(p)-f_{\alpha}(p^*)\rvert\leq\frac{\epsilon_0}{3}.
    \end{equation}
    Defining $\delta=\min\{\delta_1,\delta_2\}$ and $K$ being compact and hence totally bounded, there is a set of points $\{q_1,\ldots,q_l\}\subset K$ such that for all $p \in K$, $d(p,q(p))<\delta$ for some $q(p)\in\{q_1,\ldots,q_l\}$. Then (\ref{ppstar}) and the definition of $q(p)$ imply
    \begin{equation} \label{eps1}
    \sup_{p\in K}\lvert\hat{f}_{\alpha N}(p)-\hat{f}_{\alpha N}(q(p))\rvert<\frac{\epsilon_0}{3},
    \end{equation}
    and (\ref{ppstar2}) and the definition of $q(p)$ imply
    \begin{equation} \label{eps2}
    \sup_{p\in K}\lvert f_{\alpha}(p)-f_{\alpha}(q(p))\rvert<\frac{\epsilon_0}{3}.
    \end{equation}
    Finally, since $\{q_1,\ldots,q_l\}$ is a finite set and $\hat{f}_{\alpha N}(p)\rightarrow f_{\alpha}(p)$ on $\Omega_3$, defined as $\{\omega:\hat{f}_{\alpha N}(p_0)\rightarrow f_{\alpha}(p_0)\}\in\mathcal{F}$, there exists some $N_3(\omega_3)<\infty$ for all $\omega\in\Omega_3$ for which
    \begin{equation} \label{eps3}
    \max_{j=1,\ldots,l}\lvert\hat{f}_{\alpha N}(q_j)-f_{\alpha}(q_j)\rvert<\frac{\epsilon_0}{3}.
    \end{equation}
    So defining $N_4(\omega)=\max\{N_2(\omega),N_3(\omega)\}$ and noting that $\Omega_2$, $\Omega_3$, and hence $\Omega_4$, defined to be $\Omega_2\cap\Omega_3$, have probability 1 by Lemma \ref{l2asc},
    \begin{align*}
    &\sup_{p\in K}\lvert\hat{f}_{\alpha N}(p)-f_{\alpha}(p)\rvert \\
    \leq&\sup_{p\in K}\lvert\hat{f}_{\alpha N}(p)-\hat{f}_{\alpha N}(q(p))\rvert+\sup_{p\in K}\lvert\hat{f}_{\alpha N}(q(p))-f_{\alpha}(q(p))\rvert+\sup_{p\in K}\lvert F(q(p))-F(p)\rvert \\
    <&\frac{\epsilon_0}{3}+\max_{j=1,\ldots,l}\lvert\hat{f}_N(q_j)-F(q_j)\rvert+\frac{\epsilon_0}{3} \\
    <&\frac{\epsilon_0}{3}+\frac{\epsilon_0}{3}+\frac{\epsilon_0}{3}=\epsilon_0
    \end{align*}
    for all $\omega\in\Omega_4$, $n\geq N_4(\omega)$, using (\ref{eps1}), (\ref{eps2}) and (\ref{eps3}).
    \end{proof}
    
    \begin{lemma} \label{l2l}
    Let $\alpha\in\{1,2\}$. Recalling that $C$, defined to be the $L_\alpha$ estimator set of $y_0$, is nonempty, define $l$ to be $f_{\alpha}(p)$ for any $p\in C$, that is, $l=\min\{f_{\alpha}(q):q\in M\}$. There exist some compact $D\subset M$, $\Omega_5\in\mathcal{F}$ of probability 1 and an $N_5(\omega)<\infty$ for all $\omega\in\Omega_5$ for which $N\geq N_5(\omega)$ implies that $\hat{f}_{\alpha N}(p)>l+1$ for all $p\in M\backslash D$.
    \end{lemma}
    \begin{proof}
    For a fixed $p_0\in C$,
    \begin{align*} 
    f_{\alpha}(p)&=E(d(p,y_0)^\alpha) \nonumber \\
    &\geq E(\lvert d(p,p_0)-d(p_0,y_0)\rvert^\alpha) \nonumber\\
    &=E(d(p,p_0)^\alpha-2d(p,p_0)^{\alpha/2}d(p_0,y_0)^{\alpha/2}+d(p_0,y_0)^\alpha) \nonumber\\
    &=d(p,p_0)^\alpha-2d(p,p_0)^{\alpha/2}E(d(p_0,y_0)^{\alpha/2})+f_\alpha(p_0) \nonumber\\
    &\geq d(p,p_0)^\alpha-2d(p,p_0)^{\alpha/2}E(d(p_0,y_0)^\alpha)^{1/2}+f_\alpha(p_0) \nonumber\\
    &=f_{\alpha}(p_0)+d(p,p_0)^{\alpha/2}\big(d(p,p_0)^{\alpha/2}-2f_{\alpha}(p_0)^{1/2}\big)
    \end{align*}
    and replacing $f_{\alpha}$ and $y_0$ in the above by $\hat{f}_{\alpha N}$ and $\bar{y}_N$, an $M$-valued random element whose distribution has mass $w_{iN}$ at $y_i$ for $i=1,\ldots,N$,
    \begin{equation} \label{fnp}
    \hat{f}_{\alpha N}(p)\geq\hat{f}_{\alpha N}(p_0)+d(p,p_0)^{\alpha/2}\Big(d(p,p_0)^{\alpha/2}-2\hat{f}_{\alpha N}(p_0)^{1/2}\Big).
    \end{equation}
    There exists some sufficiently large $\Delta$ for which $d(p,p_0)>\Delta$ implies $d(p,p_0)^{\alpha/2}(d(p,p_0)^{\alpha/2}-2f_{\alpha}(p_0)^{1/2})>3$. Then define $D=\{p:d(p,C)\leq\Delta\}\supset{\{p:d(p,p_0)\leq\Delta\}}$, which is closed and bounded and hence compact by properness; $D$ might equal $M$ if $M$ is compact. By the almost sure convergence of $\hat{f}_{\alpha N}(p_0)$ to $f_{\alpha}(p_0)=l$, shown in Lemma \ref{l2asc}, there exists some $\Omega_5\in\mathcal{F}$ of probability 1 and an $N_5(\omega)<\infty$ for all $\omega\in\Omega_5$ for which $N\geq N_5(\omega)$ implies both $\hat{f}_{\alpha N}(p_0)>l-1$ and $d(p,p_0)^{\alpha/2}(d(p,p_0)^{\alpha/2}-2\hat{f}_{\alpha N}(p_0)^{1/2})>2$ for all $q\in M\backslash D$, and so, by (\ref{fnp}), that $\hat{f}_{\alpha N}(p)>l-1+2=l+1$ for all $p\in M\backslash D$. 
    \end{proof}

    We can now prove Theorem 2.
    \begin{proof}
    Define $C^\epsilon=\{p\in M:d(p,C)<\epsilon\}$. If there exist some $\theta(\epsilon)>0$, $\Omega_1\in\mathcal{F}$ of probability 1 and $N_1(\omega)<\infty$ for all $\omega\in\Omega_1$ such that, for all $N\geq N_1(\omega), \omega\in\Omega_1$,
    \begin{align} \label{final}
    &\hat{f}_{\alpha N}(p)\leq l+\frac{\theta(\epsilon)}{2}
    \end{align}
    for all $p\in C$ and
    \begin{align} \label{final2}
    &\hat{f}_{\alpha N}(p)\geq l+\theta(\epsilon)
    \end{align}
    for all $p\in M\backslash C^\epsilon$, then no minimizer of $\hat{f}_{\alpha N}$, that is, no element of the weighted sample $L_\alpha$ estimator set, is in $M\backslash C^\epsilon$ if $\omega\in\Omega_1$ and $N\geq N_1(\omega)$, proving the desired result. 
    
    The set $D_\epsilon$, defined as $D\cap(M\backslash C^\epsilon)$, is closed and bounded, being the intersection of a closed and bounded set and a closed one, and hence compact because $M$ is a proper metric space. By continuity, $f_\alpha$ attains its minimum $\min\{f_\alpha(p):p\in D_\epsilon\}$, which we call $l_\epsilon$, on the compact $D_\epsilon$, and since $D_\epsilon\cap C$ is empty, $l_\epsilon>l$, so there is some $\theta(\epsilon)\in(0,1)$ such that $l_\epsilon>l+2\theta(\epsilon)$. We now use Lemma \ref{l2sup} for $K=D$, allowing us to find a set $\Omega_4\in\mathcal{F}$ and $N_4(\omega)<\infty$ for each $\omega\in\Omega_4$ such that $N\geq N_4(\omega)$ implies that $\sup_{p\in D}\lvert\hat{f}_{\alpha N}(p)-f_{\alpha}(p)\rvert<\theta(\epsilon)/2$, and thus $\hat{f}_{\alpha N}(p)\leq f_{\alpha}(p)+\theta(\epsilon)/2=l+\theta(\epsilon)/2$ for all $p\in C\subset D$, satisfying (\ref{final}), and $\hat{f}_{\alpha N}(p)>f_{\alpha}(p)-\theta(\epsilon)/2>l_\epsilon-2\theta(\epsilon)/2>l+\theta(\epsilon)$ for all $p\in D_\epsilon\subset D$. Taking $D$, $\Omega_5$ and $N_5(\omega)$ from Lemma \ref{l2l}, which imply that $\hat{f}_{\alpha N}(p)>l+1>l+\theta(\epsilon)$ on $M\backslash D$ for all $N\geq N_5(\omega)$, $\omega\in\Omega_5$, and letting $\Omega_1=\Omega_4\cap\Omega_5$ and $N_1(\omega)=\max\{N_4(\omega),N_5(\omega)\}$, gives $\hat{f}_{\alpha N}(p)>l+\theta(\epsilon)$ for all $p\in D_\epsilon\cup(M\backslash D)\supset M\backslash C^\epsilon$ if $N\geq N_1(\omega)$ and $\omega\in\Omega_1$, satisfying (\ref{final2}), completing the proof.
   \end{proof}

   \subsection{Proof of Theorem 3} \label{appenl2sc}

    We still need a few more lemmas to prove Theorem 3.

   \begin{lemma} \label{indp}
        Let $(y_1,s_1),\ldots,(y_N,s_N):\Omega\rightarrow M\times\{1,\ldots,\Xi\}$ be independent random elements from a probability space $(\Omega,\mathcal{F},\text{pr})$ into the Cartesian product of a metric space $(M,d)$ and the set $\{1,\ldots,\Xi\}$, equipped with the product $\sigma$-algebra induced by the Borel $\sigma$-algebra $\mathcal{B}$ of $(M,d)$ and the discrete $\sigma$-algebra on $\{1,\ldots,\Xi\}$, respectively. Let $z_i\in\{0,1\}$ be a random binary variable for each positive integer $i$ and assume that $Z_N=(z_1,\ldots,z_N)$ and $(y_1,\ldots,y_N)$ are independent given $(s_1,\ldots,s_N)$. For any distinct positive integers $t_1,\ldots,t_m\leq N$, $(z_1,\ldots,z_N)$ is independent of $(y_{t_1},\ldots,y_{t_m})$ given $(s_{t_1},\ldots,s_{t_m})$.
    \end{lemma}
    \begin{proof}
        If $m=N$, the result is trivial, so assume $m<N$. Let $Y_{tm}=(y_{t_1},\ldots,y_{t_m})$, $S_{tm}=(s_{t_1},\ldots,s_{t_m})$, $s_{t'm}$ be the ordered $(N-m)$-tuple consisting of the elements of $(s_1,\ldots,s_N)$ not contained in $Y_{tm}$ and $H=\{1,\ldots,\Xi\}^{N-m}$. Then for any $U\subset\{0,1\}^N$, $V$ in the Borel $\sigma$-algebra of $M^m$ and $S\in\{1,\ldots,\Xi\}^m$,
        \begin{align*}
            &\text{pr}(Z_N\in U \mid Y_{tm}\in V,S_{tm}=S) \\
            =&\frac{\text{pr}(Z_N\in U,Y_{tm}\in V,S_{tm}=S)}{\text{pr}(Y_{tm}\in V,S_{tm}=S)} \\
            =&\frac{\sum_{S_{t'm}\in H}\text{pr}(Z_N\in U,Y_{tm}\in V,S_{tm}=S,s_{t'm}=S_{t'm})}{\text{pr}(Y_{tm}\in V,S_{tm}=S)} \\
            =&\frac{\sum_{S_{t'm}\in H}\text{pr}(Z_N\in U \mid Y_{tm}\in V,S_{tm}=S,s_{t'm}=S_{t'm})\text{pr}(Y_{tm}\in V,S_{tm}=S,s_{t'm}=S_{t'm})}{\text{pr}(Y_{tm}\in V,S_{tm}=S)} \\
            =&\frac{\sum_{S_{t'm}\in H}\text{pr}(Z_N\in U \mid S_{tm}=S,s_{t'm}=S_{t'm})\text{pr}(Y_{tm}\in V,S_{tm}=S)\text{pr}(s_{t'm}=S_{t'm})}{\text{pr}(Y_{tm}\in V,S_{tm}=S)} \\
            =&\frac{\sum_{S_{t'm}\in H}\text{pr}(Z_N\in U \mid S_{tm}=S,s_{t'm}=S_{t'm})\text{pr}(Y_{tm}\in V,S_{tm}=S)\text{pr}(s_{t'm}=S_{t'm})}{\text{pr}(Y_{tm}\in V,S_{tm}=S)} \\
            =&\sum_{S_{t'm}\in H}\text{pr}(Z_N\in U,s_{t'm}=S_{t'm}) \mid S_{tm}=S) \\
            =&\text{pr}(Z_N\in U \mid S_{tm}=S).
        \end{align*}
    \end{proof}

    For any $s\in\{1,\ldots,\Xi\}$ and $j\in\mathbb{Z}^+$, define $\Omega_T'=\{m_{TN}^s\rightarrow\infty\text{ for all $s=1,\ldots,\Xi$}\}$ and
    \begin{equation*}
    l_{Tj}=\begin{cases}
        \arg\min_{i\in\mathbb{Z}^+}\{\sum_{i=1}^kz_i=j\} &~~\mbox{on $\Omega_T'$} \\
        1 &~~\mbox{on $(\Omega_T')^c$,}
    \end{cases}
    \end{equation*}
    so that on $\Omega_T'$, $l_{Tj}$ is the $j$th positive integer for which $z_i=1$. Similarly, define $\Omega_C'=\{m_{CN}^s\rightarrow\infty\text{ for all $s=1,\ldots,\Xi$}\}$ and
    \begin{equation*}
    l_{Cj}=\begin{cases}
        \arg\min_{i\in\mathbb{Z}^+}\{\sum_{i=1}^k(1-z_i)=j\} &~~\mbox{on $\Omega_C'$} \\
        1 &~~\mbox{on $(\Omega_C')^c$.}
    \end{cases}
    \end{equation*}
    
    \begin{lemma}
    The maps $l_{Tj}$ and $l_{Cj}$ are random elements on $(\Omega,\mathcal{F},\text{pr})$, as are the maps defined by $z_{(j)}^T=z_{l_{Tj}}$, $z_{(j)}^C=z_{l_{Cj}}$, $r_{(j)}^T=r_{l_{Tj}}=r_{l_{Tj}}$, $r_{(j)}^C=r_{l_{Cj}}=r_{l_{Cj}}$, $s_{(j)}^T=s_{l_{Tj}}$ and $s_{(j)}^C=s_{l_{Cj}}$.
    \end{lemma}
    The proof is completely analogous to that of Lemma \ref{rv}.
    
    \begin{lemma} \label{iid2}
   For each $s=1,\ldots,\Xi$, the random elements $r_{(j)}^T \mid (s_{(j)}^T=s)$ and $r_T \mid (s_0=s)$ are identically distributed, as are $r_{(j)}^C \mid (s_{(j)}^C=s,\Omega_C')$ and $r_C \mid (s_0=s,\Omega_C')$, while $(r_{(1)}^T,s_{(1)}^T) \mid \Omega_T',(r_{(2)}^T,s_{(2)}^T) \mid \Omega_T',\ldots$ are independent, as are $(r_{(1)}^C,s_{(1)}^C) \mid \Omega_C',(r_{(2)}^C,s_{(2)}^C) \mid \Omega_C',\ldots$.
    \end{lemma}
    \begin{proof}
    For any positive integer $t$, $l_{Tj}=t$ is expressible as some set defined using only $(z_1,\ldots,z_t)$.
    
        Because $l_{Tj}=t$ implies that $z_t=1$, $r_{(j)}^T=r_t(z_t)=r_{Tt}$. These facts and the first paragraph of the proof of Lemma \ref{iid}, together with $(z_1,\ldots,z_t)\indep r_{Tt} \mid s_t$, which follows from Lemma \ref{indp}, and the identity of the distributions of $r_{Tt} \mid (s_t=s)$ and $r_T \mid (s_0=s)$ imply that for any $Q$ in the Borel $\sigma$-algebra $\mathcal{B}$ of $M$ and $s\in\{1,\ldots,\Xi\}$,
        \begin{align*}
            &\text{pr}(r_{(j)}^T\in Q \mid s_{(j)}^T=s) \\
            =&\frac{\text{pr}(r_{(j)}^T\in Q,s_{(j)}^T=s \mid \Omega_T')}{\text{pr}(s_{(j)}^T=s \mid \Omega_T')} \\
            =&\frac{\sum_{t=1}^\infty \text{pr}(r_{Tt}\in Q,s_t=s,l_{Tj}=t \mid \Omega_T')}{\sum_{t'=1}^\infty \text{pr}(s_{t'}=s,l_{Tj}=t' \mid \Omega_T')} \\
            =&\frac{\sum_{t=1}^\infty \text{pr}(r_{Tt}\in Q \mid s_t=s,l_{Tj}=t,\Omega_T')\text{pr}(s_t=s,l_{Tj}=t \mid \Omega_T')}{\sum_{t'=1}^\infty \text{pr}(s_{t'}=s,l_{Tj}=t' \mid \Omega_T')} \\
            =&\frac{\sum_{t=1}^\infty \text{pr}(r_{Tt}\in Q \mid s_t=s,\Omega_T')\text{pr}(s_t=s,l_{Tj}=t \mid \Omega_T')}{\sum_{t'=1}^\infty \text{pr}(s_{t'}=s,l_{Tj}=t' \mid \Omega_T')} \\
            =&\frac{\sum_{t=1}^\infty \text{pr}(r_T\in Q \mid s_0=s)\text{pr}(s_t=s,l_{Tj}=t \mid \Omega_T')}{\sum_{t'=1}^\infty \text{pr}(s_{t'}=s,l_{Tj}=t' \mid \Omega_T')} \\
            =&\text{pr}(r_T\in Q \mid s_0=s).
        \end{align*}
        The proof for $r_{(j)}^C \mid (s_{(j)}^C=s,\Omega_C')$ and $r_C \mid (s_0=s,\Omega_C')$ is analogous.
        
        This result, the first paragraphs of this proof and the proof of Lemma \ref{iid} and the assumptions that the $(r_{Ti},s_{Ti})$ are independent and the $r_{Ti} \mid (s_{Ti}=s)$ are distributed identically to $r_T \mid (s_0=s)$ imply that for any positive integer $m$, any $j_1,\ldots,j_m\in\mathbb{Z}^+$ such that $j_1<\cdots<j_m$, any $Q_1,\ldots,Q_m$ in the Borel $\sigma$-algebra $\mathcal{B}$ of $M$ and any $s_1,\ldots,s_m\in\{1,\ldots,\Xi\}$,
        \begin{align*}
            &\text{pr}\bigg(\bigcap\limits_{l=1}^m\{(r_{(j_1)}^T,s_{(j_l)}^T)\in (Q_l,\{s_l\})\}\,\,\bigg|\,\,\Omega_T'\bigg) \\
            =&\sum_{t_1<\cdots<t_m}\text{pr}\bigg(\bigcap\limits_{l=1}^m\{r_{Tt_l}\in Q_l,s_{t_l}=s_l,l_{Tj_l}=t_l\}\,\,\bigg|\,\,\Omega_T'\bigg) \\
            =&\sum_{t_1<\cdots<t_m}\text{pr}\bigg(\bigcap\limits_{l=1}^m\{r_{Tt_l}\in Q_l\}\,\,\bigg|\,\,\bigcap\limits_{l=1}^m\{s_{t_l}=s_l,l_{Tj_l}=t_l\}\bigcap\Omega_T'\bigg) \\
            &\qquad\text{pr}\bigg(\bigcap\limits_{l=1}^m\{s_{t_l}=s_l,l_{Tj_l}=t_l\}\,\,\bigg|\,\,\Omega_T'\bigg) \\
            =&\sum_{t_1<\cdots<t_m}\text{pr}\bigg(\bigcap\limits_{l=1}^m\{r_{Tt_l}\in Q_l\}\,\,\bigg|\,\,\bigcap\limits_{l=1}^m\{s_{t_l}=s_l\}\bigcap\Omega_T'\bigg)\text{pr}\bigg(\bigcap\limits_{l=1}^m\{s_{t_l}=s_l,l_{Tj_l}=t_l\}\,\,\bigg|\,\,\Omega_T'\bigg) \\
            =&\sum_{t_1<\cdots<t_m}\frac{\text{pr}\big(\bigcap\limits_{l=1}^m\{r_{Tt_l}\in Q_l,s_{t_l}=s_l\}\Big)}{\text{pr}\Big(\bigcap\limits_{l=1}^m\{s_{t_l}=s_l\Big)}\text{pr}\bigg(\bigcap\limits_{l=1}^m\{s_{t_l}=s_l,l_{Tj_l}=t_l\}\,\,\bigg|\,\,\Omega_T'\bigg) \\
            =&\sum_{t_1<\cdots<t_m}\Bigg(\prod_{l=1}^m\frac{\text{pr}(r_{Tt_l}\in Q_l,s_{t_l}=s_l)}{\text{pr}(s_{t_l}=s_l)}\Bigg)\text{pr}\bigg(\bigcap\limits_{l=1}^m\{s_{t_l}=s_l,l_{Tj_l}=t_l\}\,\,\bigg|\,\,\Omega_T'\bigg) \\
            =&\sum_{t_1<\cdots<t_m}\Bigg(\prod_{l=1}^m\text{pr}(r_{Tt_l}\in Q_l \mid s_{t_l}=s_l)\Bigg)\text{pr}\bigg(\bigcap\limits_{l=1}^m\{s_{t_l}=s_l,l_{Tj_l}=t_l\}\,\,\bigg|\,\,\Omega_T'\bigg) \\
            =&\sum_{t_1<\cdots<t_m}\bigg(\prod_{l=1}^m\text{pr}(r_T\in Q_l \mid s_0=s_l)\bigg)\text{pr}\bigg(\bigcap\limits_{l=1}^m\{s_{t_l}=s_l,l_{Tj_l}=t_l\}\,\,\bigg|\,\,\Omega_T'\bigg) \\
            =&\bigg(\prod_{l=1}^m\text{pr}(r_T\in Q_l \mid s_0=s_l)\bigg)\sum_{t_1<\cdots<t_m}\text{pr}\bigg(\bigcap\limits_{l=1}^m\{s_{t_l}=s_l,l_{Tj_l}=t_l\}\,\,\bigg|\,\,\Omega_T'\bigg) \\
            =&\bigg(\prod_{l=1}^m\text{pr}(r_{(j_l)}^T\in Q_l \mid s_{(j_l)}^T=s_l)\bigg)\text{pr}\bigg(\bigcap\limits_{l=1}^m\{s_{(j_l)}^T=s_l\}\,\,\bigg|\,\,\Omega_T'\bigg) \\
            =&\prod_{l=1}^m \text{pr}\left((r_{(j_l)}^T,s_{(j_l)}^T)\in(Q_l,\{s_l\}) \,\,\bigg|\,\, \Omega_T'\right).
        \end{align*}
        The result follows since $\{(Q,s) \mid Q\in\mathcal{B},s=1,\ldots,\Xi\}$ generates the induced product $\sigma$-algebra of $M\times\{1,\ldots,\Xi\}$. The proof for $(r_{(1)}^C,s_{(1)}^C) \mid \Omega_C',(r_{(2)}^C,s_{(2)}^C) \mid \Omega_C',\ldots$ is analogous.
    \end{proof}

    We can now prove Theorem 3.
    
    \begin{proof}

    By Lemma \ref{iid2}, $(r_T,S),(r_{(1)}^T,s_{(1)}^T),(r_{(2)}^T,s_{(2)}^T)$ satisfy the conditions necessary to apply Theorem 2, and so do $(r_C,S),(r_{(1)}^C,s_{(1)}^C),(r_{(2)}^C,s_{(2)}^C)$. Then, comparing the weights $w_{iN}$ in Theorem 2 to $w_{TiN}$ and $w_{CiN}$ in the definition of $T_{\alpha}$, we see that $\bar{A}_{TN}^\alpha$ and $\bar{A}_{CN}^\alpha$ are the stratification-weighted sample $L_\alpha$ estimator sets for the treated and control groups, respectively. Since $L_\alpha$ estimator sets are compact by Lemma \ref{l2basic} and $d:M\times M\rightarrow\mathbb{R}$ is continuous, for each $N$ there exist some $\bar{r}_{T}^*\in\bar{A}_{TN}^\alpha$ and $\bar{r}_{C}^*\in\bar{A}_{CN}^\alpha$ that achieve the infimum. Let $\mu_{\alpha T}$ and $\mu_{\alpha C}$ be the unique $L_\alpha$ estimators of $r_T$ and $r_C$, respectively. Then by Theorem 2, for any $\epsilon>0$, there exist some $\Omega_T'\in\mathcal{F}$ for which $\text{pr}(\Omega_T')=1$ and $N_T'(\omega)<\infty$ for all $\omega\in\Omega_T'$ such that $d(\bar{r}_{T}^*,\mu_{\alpha T})<\epsilon/2$ for all $N\geq N_T'(\omega)$, and some $\Omega_C'\in\mathcal{F}$ for which $\text{pr}(\Omega_C')=1$ and $N_C'(\omega)<\infty$ for all $\omega\in\Omega_C'$ such that $d(\bar{r}_{C}^*,\mu_{\alpha C})<\epsilon/2$ for all $N\geq N_C'(\omega)$. Letting $\Omega'=\Omega_T'\bigcap\Omega_C'$ and $N'(\omega)=\max\{N_T'(\omega),N_C'(\omega)\}$ for all $\omega\in\Omega'$, $T_{\alpha}(Z_N,R_N,S_N)=d(\bar{r}_{T}^*,\bar{r}_{C}^*)\in(d(\mu_{\alpha T},\mu_{\alpha C})-\epsilon/2-\epsilon/2,d(\mu_{\alpha T},\mu_{\alpha C})+\epsilon/2+\epsilon/2)=(d(\mu_{\alpha T},\mu_{\alpha C})-\epsilon,d(\mu_{\alpha T},\mu_{\alpha C})+\epsilon)$ for all $N\geq N'(\omega)$ and $\text{pr}(\Omega')=1$, proving the desired result for $T_{\alpha}$.
   \end{proof}

    \subsection{Proof of Proposition 1} \label{boot}

    \begin{proof}
    Using the unconfoundedness from (4) in the main paper for $N=1$, 
    \begin{align*} E_F[E_F[\rho(r_1)\mid z_1=1,s_1]]&=E_F[E_F[\rho(r_{T1})\mid z_1=1,s_1]] \\
    &=E_F[E_F[\rho(r_{T1})\mid s_1]] \\
    &=E_F[\rho(r_{T1})],
    \end{align*}
    where $\rho:M\rightarrow\mathbb{R}$ is any measurable function, and similarly,  
    \begin{align*}
        E_F[E_F[\rho(r_1)\mid z_1=0,s_1]]=E_F[\rho(r_{C1})].
    \end{align*}
       Counterintuitively, the two equations above imply that assuming unconfoundedness, it is possible to calculate the distributions of $r_{T1}$ and $r_{C1}$ just from the distribution of $(z_1,r_1,s_1)$ without direct knowledge of the distributions of $r_{T1}$ and $r_{C1}$. Then letting $\rho(x)=d(p,x)^\alpha$ and by the uniqueness of the $L_\alpha$ estimators $\mu_T$ and $\mu_C$ of $r_T$ and $r_C$, respectively,
    \begin{equation} \label{gl2}
    H_{\alpha}(F)=d(\mu_{\alpha T},\mu_{\alpha C})
    \end{equation}
    is $\textsc{aate}$ if $\alpha=2$ and $\textsc{amte}$ $\alpha=1$.

    Now 
    \begin{align*}
    E_{\hat{F}_N}[E_{\hat{F}_N}[\rho(r_1)\mid z_1=1,s_1]]&=\sum_{s=1}^\Xi E_{\hat{F}_N}[I\{s_1=s\}]E_{\hat{F}_N}[\rho(r_1)\mid z_1=1,s_1=s] \\
    &=\sum_{s=1}^\Xi \frac{m_N^s}{N}\sum_{i=1}^N\frac{z_iI\{s_i=s\}\rho(r_i)}{m_{TN}^s} \\
    &=\sum_{i=1}^N w_{TiN}\rho(r_i),
    \end{align*}
    where each stratum-wise weight $\hat{\lambda}_N^s$ is set to be $m_{N}^s/N$. Similarly,
    \begin{align*}
    E_{\hat{F}_N}[E_{\hat{F}_N}[\rho(r_1)\mid z_1=0,s_1]]=\sum_{i=1}^N w_{CiN}\rho(r_i).
    \end{align*}
    Then again letting $\rho(x)=d(p,x)^\alpha$, 
    \begin{equation} \label{gl2hat}
    H_{\alpha}(\hat{F}_N)=T_{\alpha}(Z_N,R_N,S_N).
    \end{equation}
    \end{proof}
    
\end{document}